\documentclass[10pt]{IEEEtran}
\onecolumn
\pdfoutput=1
\def\BibTeX{{\rm B\kern-.05em{\sc i\kern-.025em b}\kern-.08em
    T\kern-.1667em\lower.7ex\hbox{E}\kern-.125emX}}
\usepackage{graphicx}
\usepackage{color}
\newtheorem{theorem}{Theorem}

\newtheorem{example}{Example}

\newtheorem{note}{Note}

\newcommand{\1}{{\bf 1}} 
\newcommand{\E}{\mathsf{E}} 

\newcommand{\card}[1]           {\left| #1\right|}

\newcommand{\bea}{\begin{eqnarray}}
\newcommand{\eea}{\end{eqnarray}}
\newcommand{\beas}{\begin{eqnarray*}}
\newcommand{\eeas}{\end{eqnarray*}}

\begin{document}
\title{Probing Capacity}
\author{\IEEEauthorblockN{Himanshu Asnani\IEEEauthorrefmark{1}, Haim
Permuter\IEEEauthorrefmark{2} and 
Tsachy Weissman\IEEEauthorrefmark{3}} 
\thanks{\IEEEauthorblockA{\IEEEauthorrefmark{1}Stanford University, Email:
asnani@stanford.edu.}} 
\thanks{\IEEEauthorblockA{\IEEEauthorrefmark{2}Ben Gurion University, Email:
haimp@bgu.ac.il.}} 
\thanks{\IEEEauthorblockA{\IEEEauthorrefmark{3}Stanford University, Email:
tsachy@stanford.edu.}}}

\maketitle


\maketitle

\begin{abstract}
We consider the problem of optimal probing of states of a channel by
transmitter and receiver for maximizing rate of reliable communication. The
channel is discrete
memoryless (DMC) with i.i.d. states. The encoder
takes probing actions dependent on the message. It then uses the state
information obtained from probing causally
or
non-causally to generate channel input symbols. The decoder may also take
channel probing actions as a function of the observed channel output and use the
channel state information thus acquired, along with the channel output, to
estimate the message. We refer to the maximum
achievable rate	
for reliable communication for such systems as the
\textquoteleft \textit{Probing Capacity}\textquoteright. We
characterize this capacity when the encoder and decoder actions are cost
constrained. To motivate the problem, we begin by characterizing
 the trade-off between the capacity and fraction of channel states the encoder
is allowed to observe, while the
decoder 
is aware of channel states. In this setting of
\textquoteleft\textit{to observe or not to observe}\textquoteright\ state
at the encoder, we compute certain numerical examples and note a pleasing
phenomenon, where encoder can observe a relatively small fraction of states and
yet communicate at maximum rate, i.e. rate when observing states at encoder is
not cost constrained.	
\end{abstract}


\section{Introduction}
Shannon showed the importance of availability of channel state at the encoder
for communication system
in his seminal paper \cite{ShannonChannel}, where he computed capacity of DMC
with i.i.d. states 
available causally to the encoder. This spawned an active research in the area
of
channel coding and was extended to various scenarios, notably 
for storage in computer memory. Kuznetsov and Tsybakov in
\cite{KuznetsovTsybakov}
constructed defect-correcting codes for coding in computer memory with defective
cells. Gel'fand and Pinsker in \cite{GelfandPinsker}, extended work in
\cite{ShannonChannel} to the case where channel states are available
non-causally to the encoder, again with applications for computer memories,
which was further 
researched by Heegard and El Gamal in \cite{HeegardGamal}.
 Keshet, Steinberg and Merhav presented a detailed survey in
\cite{KeshetSteinbergMerhav} on \textit{channel coding in the presence of state
information}, where 
the channel state information (CSI) signal is available at the transmitter
(CSIT) or at the 
receiver (CSIR), or both.\par
Permuter and Weissman introduced the notion of
\textit{actions} in source coding context in \cite{HaimTsachyVendor}. Their
setting is a generalization of the Wyner-Ziv
source coding 
with decoder side information problem (\cite{WynerZiv}), where now the decoder
can take actions based on the index
obtained from the encoder 
to affect the formation or availability of side information. Weissman, in
\cite{TsachyChannel}, studied the channel coding dual where the transmitter
takes actions that affect the formation of channel states. This framework
captures 
various new coding scenarios which include two stage
recording on a memory with defects, motivated by similar problems in
magnetic recording and computer memories. Kittichokechai \textit{et al} in
\cite{Kittichokechai} studied 
a variant of the problem in \cite{HaimTsachyVendor} and \cite{TsachyChannel},
where
encoder and decoder both have action dependent partial side information.
However, in the source coding
formulation of \cite{HaimTsachyVendor}, they 
restricted the actions 
to be taken by decoder while in the channel coding scenario of
\cite{TsachyChannel} and \cite{Kittichokechai}, actions were taken only
by the encoder.
\par
In this paper, we revisit channel coding scenarios but now cost constrained
actions are taken to
acquire any partial or complete
channel state information by the encoder, the decoder or both. 
Our framework is aimed at capturing and understanding the trade offs involved in
natural scenarios where the acquisition of channel state information is
associated with expenditure of costly system resources. The encoder and
decoder actions are cost constrained creating tension between achievable rate
and the cost of acquisition of the channel state (or the defect) information.
Note that our framework differs from those of 
\cite{TsachyChannel} and
\cite{Kittichokechai} where actions
affect the channel, followed by channel encoding. In our scenario
channel
statistics are not affected, i.e., nature generates the state sequence i.i.d
$\sim P_S$. Our work is novel in the sense that not only the encoder but
the decoder also takes actions to acquire channel state information. Encoder
takes actions ($A_e$) depending on messages. Decoder also
takes actions ($A_d$) depending upon
observed channel output. Using their respective actions, encoder and decoder
observe partial states, $S_e$ and $S_d$ through discrete memoryless channel
(DMC), $P_{S_e,S_d|S,A_e,A_d}$. The encoder can causally or non-causally use its
partial state information to generate the
channel input symbols. In this paper, we characterize the fundamental limit of
such a framework and call it
\textbf{Probing Capacity}. When the actions are not taken by the decoder, there
is 
an equivalence between our setting and that of channels with action dependent
states as in
\cite{TsachyChannel}, which we make explicit in Section
\ref{equivalence}.
\par
The rest of the paper is organized as follows. We begin with a motivating
scenario in Section \ref{sample}, where decoder knows the complete state and
the encoder takes message dependent binary actions 
 \textit{to observe or not to observe} the channel state. This is
generalized in Section
\ref{equivalence}, when only encoder takes actions. This section also 
establishes the equivalence
between our framework of optimal probing and that of channels with action
dependent states in \cite{TsachyChannel}. Motivated by the framework of
communication over slow
fading channels, where the information of channel states is to be
exploited on the fly, we have in Section
\ref{causal}
characterization of the \textit{probing capacity} where encoder takes actions to
get
channel states and
use them causally to construct channel inputs and decoder takes actions
strictly causally dependent on channel outputs. Note that in this section,
we characterize a novel and a generalized setting, where both encoder and
decoder take costly actions to get channel state information. Later in this
section,
inspired by coding
on computer memory with defects, we explain the
non-causal case, i.e., when channel
states are used 
non-causally by the encoder to generate channel input symbols and decoder
waits for the entire channel output before taking actions to get channel
states. This in general is a hard problem and we show its equivalence
to a standard relay
channel with infinite lookahead. In Section \ref{example}, we
work out several examples, with some surprising implications. The paper is
concluded in Section \ref{conclusion} with directions of future research. 

\section{To Observe or Not to Observe Channel States at Encoder}
\label{sample}
We begin by explaining the notation to be used throughout this paper.
Let upper case, lower case, and calligraphic letter denote, respectively, random
variables, specific or deterministic values which random variables may assume,
and
their alphabets. For two jointly distributed random variables, $X$ and $Y$, let
$P_X$, $P_{XY}$ and $P_{X|Y}$ respectively denote the marginal of $X$, joint
distribution of $(X,Y)$ and conditional distribution of 
$X$ given $Y$. $X_{m}^{n}$ is a shorthand for $n-m+1$ tuple
$\{X_m,X_{m+1},\cdots,X_{n-1},X_n\}$. We impose the assumption of finiteness of
cardinality on all alphabets, unless otherwise indicated.\par
In this section, we consider the problem of optimal probing where encoder takes
a
\textquoteleft costly\textquoteright \ action depending upon message and use
it to probe the channel and
observe or
not 
the channel state. The actions are binary, hence while action, $A=1$
corresponds to the case when encoder observes the channel state, action, $A=0$
implies no acquired state information. Note that such a kind of abstraction
taps in the motivation considered in Compressed Sensing framework in
\cite{Donoho}, where due to cost of sensing and measurement, you aim to observe
only a few noisy signal observations and construct the original signal
accurately. We further assume decoder knows 
the complete state information and that the encoder uses partial state
information
non-causally to generate channel input symbol. 
\subsection{Problem Setup}
The setting is depicted in Figure \ref{EncoderStateEstLogicD}: Message
$M$ is selected uniformly from a uniform distribution on the message set
$\mathcal{M}=\{1,2,\cdots,\card{\mathcal{M}}\}$. Nature generates states
sequence
$S^n\in \mathcal{S}^n$ i.i.d $\sim P_S$, independent of message.
A $(2^{nR},n)$ code consists of :
\begin{itemize}
\item \textit{Probing Logic} : $f_A:M\rightarrow A^n\in \{0,1\}^n$ such that the
action sequence $A^n$ satisfies the cost constraints
\bea\label{costconstraints}
\Lambda(A^n)=\frac{1}{n}\sum_{i=1}^{n}\Lambda(A_i)\leq \Gamma,
\eea
where $\Lambda(\cdot)$ is the cost function while $\Gamma$ is the cost
constraint. Given nature generated state sequence $S^n$ and
message dependent
action sequence $A^n$, encoder receives partial state information $S_e^n\in
\{\{\ast\}\cup\mathcal{S}\}^n$ through a deterministic channel characterized
by,
\bea
S_e=h(S,A)=S \mbox { if }A=1,\\
S_e=h(S,A)=\ast \mbox { if }A=0,
\eea
where $\ast$ stands for erasure or no information of state symbol. Thus, $A=1$
corresponds to an observation of the channel state while $A=0$ to a lack of an
observation. Without loss of generality we can assume, $\Gamma(0)=0$.
\item \textit{Encoding} : $f_e : (M,S^n_e)\rightarrow X^n\in\mathcal{X}^n$, i.e.
encoder uses the partial state information 
non-causally to generate channel input symbols.
\item \textit{Decoding} : $f_d : (Y^n,S^n)\rightarrow \hat{M}\in\mathcal
\{1,2,\cdots,\card{\mathcal{M}}\}$, 
where the channel output $Y^n\in\mathcal{Y}^n$. 
\end{itemize}
The joint PMF on $(M,A^n,S^n,S^n_e,X^n,Y^n,\hat{M})$ induced by a given scheme
is
\bea\label{mainjoint}
&&P_{M,A^n,S^n,S^n_e,X^n,Y^n,\hat{M}}(m,a^n,s^n,s^n_e,x^n,y^n,\hat{m}
)\nonumber\\
&=&\frac{\1_{\{f_A(m)=a^n,f_e(m,s_e^n)=x^n,f_d(y^n,s^n)=\hat{m}\}}}{\card{
\mathcal{M}}}\prod_{i=1}^{n}P_S(s_i)\1_{\{s_
{e,i}=h(s_i,a_i)\}}P_{Y|X,S}(y_i|x_i,s_i).
\eea
The probability of error is calculated as $P_e=P(M\neq \hat{M}(Y^n,S^n))$. The
rate
$R$ is said to be achievable if there exists a sequence of $(2^{nR},n)$ codes
for
increasing block lengths satisfying the cost constraints
(\ref{costconstraints}) with $\frac{1}{n}\log
\card{\mathcal{M}}\leq R$ and
$P^n_e\stackrel{n\rightarrow\infty}{\longrightarrow}0$

\begin{figure}[htbp]
\begin{center}
\scalebox{0.7}{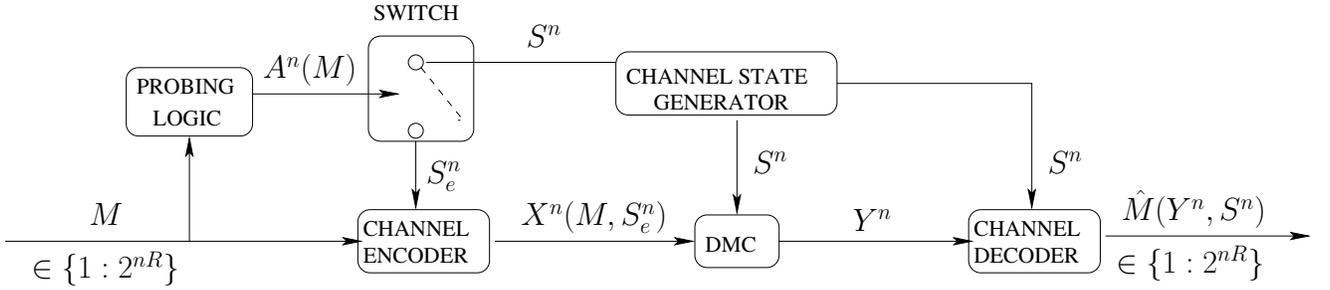}
\caption{Encoder takes message dependent actions to observe state, encodes
using available partial state information non-causally
while decoder knows the complete channel state sequence.}
\label{EncoderStateEstLogicD}
\end{center}
\end{figure}
\subsection{Probing Capacity}
\begin{theorem}
\label{theorem1}
The cost constrained \textit{\textquoteleft probing capacity\textquoteright} of
the system in Fig.
\ref{EncoderStateEstLogicD} with 
channel inputs constructed using the observed state sequence non-causally while 
decoder has complete information of the state is given
by
\bea\label{cap1}
C(\Gamma)=\max[I(X;Y|S)],
\eea
where maximization is over all joint distributions of the form
\bea
P_{A,S,S_e,X,Y}=P_AP_S\1_{\{S_e=h(S,A)\}}P_{X|S_e,A}P_{Y|X,S},
\eea
for some $P_A,P_{X|S_e,A}$ such that $E[\Lambda(A)]\leq\Gamma$.
\end{theorem}

\begin{proof}
 \label{proof1}
\newline
 \textit{Achievability :} We use \textit{Rate-Splitting} and
\textit{Multiplexing} to achieve
capacity (for a similar scheme refer to \cite{GoldsmithVaraiya}). Note that in
this
problem while 
knowing $S_e$ we know $A$, hence we would
show 
achievability with $P_{X|S_e,A}$ replaced by $P_{X|S_e}$. Without loss of
generality
we assume $\mathcal{S}=\{1,2,\cdots,\card{\mathcal{S}}\}$, hence
$\mathcal{S}_e=\{\ast,1,2,\cdots,\card{\mathcal{S}}\}$.
 Fix $P_A,P_{X|S_e}$ which achieve $C(\frac{\Gamma}{1+\epsilon})$. We
\textit{split} 
message $M$ of rate $R$ into two messages
$M_1$ and $M_2$ of rate $R_1$ and $R_2$ respectively.
\begin{itemize}
 \item \textit{Generation of Codebooks : }
\begin{itemize}
 \item Generate codebook $\mathcal{C}_A$ of $\{A^n(m_1)\}_{m_1=1}^{2^{nR_1}}$
$n$-tuples
i.i.d. $\sim P_A$. To send message $M=(M_1,M_2)$, if $A^n(M_1)\in
T^n_\epsilon(A)$ ($T_\epsilon^n$ are
typical in the sense of \cite{CsiszarKorner}), then action $A^n(M_1)$ is
taken, else $A_0^n=(0,0,\cdots,0)$ is taken. If $A^n(M_1)\in
T^n_\epsilon(A)$, then by typical average lemma
$\cite{GamalKim}$, constraints are satisfied as,
\bea
\Lambda(A^n)=\frac{1}{n}\sum_{i=1}^{n}\Lambda(A_i)\leq
(1+\epsilon)E(\Lambda(A))=\Gamma.
\eea
\item For every $A^{n}(m_1)$, generate a codebook $\mathcal{C}_X(m_1)$ of 
$\{(X_e^n(m_1,m_2),X_1^n(m_1,m_2),\cdots,X_{\card{\mathcal{S}}}^n(m_1,m_2))\}_{
m_2=1}^{2^{nR_2}}$ \newline $(\card{\mathcal{S}}+1)n$-tuples
such that $X_e^n$, $X_1^n,\cdots, X_{\card{\mathcal{S}}}^n$ are i.i.d. $\sim
P_{X|S_e=\ast}, P_{X|S_e=1}, 
\cdots, P_{X|S_e=\card{\mathcal{S}}}$ respectively. Also generate a codebook
$\mathcal{C}^0_X$ of codewords $\{(X_{0,e}^n(m_2)\}_{
m_2=1}^{2^{nR_2}}$ i.i.d. $\sim
P_{X|S_e=\ast}$.
\end{itemize}
 \item \textit{Encoding :} 
\begin{itemize}
\item Given a message $M=(M_1,M_2)$, encoder decides to take
actions $A^n(M_1)$ or $A^n_0$ depending whether $A^n(M_1)$ is in
$T^n_\epsilon(A)$ or not. If $A^n(M_1)\in
T^n_\epsilon(A)$ encoder finds
$S_e^n=h(A^n(M_1),S^n)$, and then sends $X^n(M_1,M_2)$ using the 
following \textit{multiplexing}.
\bea
X_i&=&X_{e,i}(M_1,M_2) \mbox{ if } S_{e,i}=\ast,\\
X_i&=&X_{j,i}(M_1,M_2) \mbox{ if }
S_{e,i}=j\in\{1,2,\cdots,\card{\mathcal{S}}\}.\nonumber\\
\eea
If $A^n(M_1)\notin T^n_\epsilon(A)$, encoder sends $X^n_0(M_2)$.
\end{itemize}
 \item \textit{Decoding} : We perform \textit{Successive
Decoding} and \textit{Demultiplexing}. By successive decoding we mean that
actions are decoded first by decoder and then the actual codewords.
\begin{itemize}
 \item On obtaining the channel output sequence $Y^n$ and channel state sequence
$S^n$ decoder finds the
smallest value of $\hat{M}_1$ for which 
$(A^{n}(\hat{M}_1),Y^n,S^n)\in T_\epsilon^n(A,Y,S)$. If there is
no such $\hat{M_1}$, decoder assumes $\hat{M}_1=1$. 
 \item Once the decoder decodes the value of $M_1$, if $A^n(M_1)\in
T^n_\epsilon(A)$, it knows
$S^n_{e}=h(A^n(M_1),S^n)$ and hence,
using the codebook $C_X(M_1)$, it \textit{demultiplexes}
$\{(X_e^n(M_1,m_2),X_1^n(M_1,m_2),\cdots,X_{\card{\mathcal{S}}}^n(M_1,m_2))\}_{
m_2=1}^{2^{nR_2}}$
 to construct $X^n(M_1,m_2)_{m_2=1}^{2^{nR_2}}$ sequences as,
\bea
X_i(M_1,m_2)&=&X_{e,i}(M_1,m_2) \mbox{ if } S_{e,i}=\ast,\\
X_i(M_1,m_2)&=&X_{j,i}(M_1,m_2) \mbox{ if }
S_{e,i}=j\in\{1,2,\cdots,\card{\mathcal{S}}\}.
\eea
 \item After demultiplexing, if $A^n(M_1)\in
T^n_\epsilon(A)$, decoder finds the
smallest value of $\hat{M_2}$ for which 
$(X^{n}(M_1,\hat{M}_2),Y^n|S^n,A^{n}(M_1))\in T_\epsilon^n(X,Y|S^n,A^n(M_1)$. If
there is
no such $\hat{M_2}$, decoder assumes $\hat{M_2}=1$. If $A^n(M_1)\notin
T^n_\epsilon(A)$, decoder finds the
smallest value of $\hat{M_2}$ for which 
$(X_0^{n}(\hat{M}_2),Y^n|S^n)\in T_\epsilon^n(X,Y|S^n)$, else
$\hat{M_2}=1$ is assumed.
\end{itemize}
\item \textit{Analysis of Probability of Error :} Without loss of generality
we can assume $M=(M_1,M_2)=(1,1)$ was sent. We have the following error
events, 
\begin{itemize}
 \item $\mathcal{E}_{11}=\{A^n(1),Y^n,S^n\}\notin T_\epsilon^n(A,Y,S)$.
 \item $\mathcal{E}_{12}=\{A^n(\hat{m}_1),Y^n,S^n\}\in T_\epsilon^n(A,Y,S)$
for $\hat{m}_1\neq 1$.
 \item $\mathcal{E}_{21}=\{X^n(1,1),Y^n|S^n,A^{n}(1)\}\notin
T_\epsilon^n(X,Y|S^n,A^n(1))$.
 \item $\mathcal{E}_{22}=\{X^n(1,\hat{m}_2),Y^n|S^n,A^n(1)\}\in
T_\epsilon^n(X,Y|S^n,A^n(1))$
for $\hat{m}_2\neq 1$.
\end{itemize}
Let $\mathcal{E}_0=P((A^n(1),X^n(1,1))\in T^n_\epsilon(A,X))$. Hence,
\bea
P(\mathcal{E})&=&P(\mathcal{E}_0\cap\mathcal{E})+P(\mathcal{E}^c_0\cap\mathcal{E
} )\\
&\leq&P(\mathcal{E}_0\cap\mathcal{E})+P(\mathcal{E}^c_0).
\eea
Note that by LLN (\cite{GamalKim}), $P(\mathcal{E}^c_0)\rightarrow 0$ as
$n\rightarrow\infty$. \newline
We will now show that $P(\mathcal{E}_0\cap\mathcal{E})\rightarrow0$. 
Let $\mathcal{E}_1=\mathcal{E}_{11}\cup\mathcal{E}_{12}$ and
$\mathcal{E}_2=\mathcal{E}_{21}\cup\mathcal{E}_{22}$.
By Law of Large Numbers, (LLN, (\cite{GamalKim}),
$P(\mathcal{E}_0\cap\mathcal{E}_{11})\rightarrow0$. By Packing Lemma
(\cite{GamalKim}),
$P(\mathcal{E}_0\cap\mathcal{E}_{12})\rightarrow0$ if
$R_1<I(A;Y,S)=I(A;Y|S)$ which
implies by union bound 
$P(\mathcal{E}_0\cap\mathcal{E}_1)\leq P(\mathcal{E}_0\cap\mathcal{E}_{11}
)+P(\mathcal{E}_0\cap\mathcal{E}_{12})\rightarrow 0$. \newline
Similarly by LLN,
$P(\mathcal{E}_0\cap\mathcal{E}^c_{1}\cap\mathcal{E}_{21})\rightarrow0$ and by
Packing Lemma
$P(\mathcal{E}_0\cap\mathcal{E}^c_{1}\cap\mathcal{E}_{22}
)\rightarrow0$ if
$R_2<I(X;Y|S,A)$ which implies by the union bound
$P(\mathcal{E}_0\cap\mathcal{E}^c_{1}\cap\mathcal{E}_2)\leq
P(\mathcal{E}_0\cap\mathcal{E}^c_{1} \cap\mathcal{E}_{21})
+P(\mathcal{E}_0\cap\mathcal{E}^c_{1}\cap\mathcal{E}_{22})\rightarrow 0$.
Hence the total probability of error 
\bea
P(\mathcal{E}_0\cap\mathcal{E})=P(\mathcal{E}_0\cap(\mathcal{E}_1\cup\mathcal{E}
_2))\leq
P(\mathcal{E}_0\cap\mathcal{E}_1)+P(\mathcal{E}_0\cap\mathcal{E}^c_1\cap\mathcal
{E}_2)\rightarrow 0, 
\eea
if $R_1< I(A;Y|S)$ and $R_2< I(X;Y|S,A)$. Therefore we obtain for vanishing
probability of error that 
\bea
R&=&R_1+R_2\\
&<&I(A;Y|S)+I(X;Y|A,S)\\
&=&I(X,A;Y|S)\\
&=&H(Y|S)-H(Y|X,S,A)\\
&=&I(X;Y|S)=C(\frac{
\Gamma}{1+\epsilon}).
\eea
Proof of achievability is completed by taking $\epsilon\rightarrow 0$.
\end{itemize}
 \textit{Converse :} Suppose rate $R$ is achievable. Now consider a sequence of
$(2^{nR},n)$ codes for which we have 
$P^n_e\stackrel{n\rightarrow\infty}{\longrightarrow}0$. Consider
\bea
nR&=&H(M)\\
&\stackrel{(a)}{=}&H(M|S^n)\\
&=&I(M;Y^n|S^n)+H(M|Y^n,S^n).\label{mutual1}
\eea
By Fano's Inequality (\cite{CoverThomas})
\bea
\label{fano1}H(M|Y^n,S^n)\leq 1+P^n_eR\leq n\epsilon_n,
\eea
where $\epsilon_n\stackrel{n\rightarrow\infty}{\longrightarrow}0$. Now Consider
\bea
I(M;Y^n|S^n)&=&H(Y^n|S^n)-H(Y^n|M,S^n)\\
&\stackrel{(b)}{=}&\sum_{i=1}^{n}H(Y_i|S^n,Y^{i-1})\nonumber\\
&-&\sum_{i=1}^{n}
H(Y_i|Y^ { i-1 } ,
M,S^n,A^n,S^n_e,X^n)\nonumber\\\\
&\stackrel{(c)}{\leq}&\sum_{i=1}^{n}H(Y_i|S_i)-\sum_{i=1}^{n}H(Y_i|X_i,S_i)\\
&=&\sum_{i=1}^{n} I(X_i;Y_i|S_i)\\
&\stackrel{}{\leq}&\sum_{i=1}^{n} C(\Lambda(A_i))\\
&\stackrel{(d)}{\leq}& n C(\frac{1}{n}\sum_{i=1}^{n}\Lambda(A_i))\\
&=& n C(\Lambda(A^n))\\
&\stackrel{(e)}{\leq}&n C(\Gamma),\label{costcap1}
\eea
where 
\begin{itemize}
 \item (a) follows from the fact that message is independent of state sequence.
 \item (b) follows from the fact that $A^n=A^n(M)$, $S_e=h(S,A)$ and
$X^n=X^n(M,S_e^n)$.
 \item (c) follows from the fact that conditioning reduces entropy and from the
markov chain, $Y_i-(X_i,S_i)-(Y^{ i-1 },M,S^{n\backslash
i},A^n,S^n_e,X^{n\backslash i})$ which is due to the induced joint probability
distribution as in Eq. (\ref{mainjoint}).
 \item (d) follows from the fact that $C(\Gamma)$ is concave in $\Gamma$. This
is proved
as follows. Let $C(\Gamma_1)$ and $C(\Gamma_2)$ be respectively achieved at
joint $P^1_{A}P^1_{X|S_e,A}$ and $P^2_{A}P^2_{X|S_e,A}$. Let $P^1(\cdot)$ and
$P^2(\cdot)$ be the corresponding joint distributions. Since $C(\Gamma)$ is
nondecreasing in $\Gamma$, therefore we have
\bea
\E_{P^1}[\Lambda(A)]&=&\Gamma_1\\
\E_{P^2}[\Lambda(A)]&=&\Gamma_2.
\eea
Now consider a joint distribution $P^{\lambda}=\lambda P^1+(1-\lambda)P^2$.
Clearly
\bea
\E_{P^{\lambda}}[\Lambda(A)]=\lambda \Gamma_1+(1-\lambda)\Gamma_2.
\eea
Now observe that $I(X;Y|S)$ is concave in $P(Y|S)$ which is linear in
$P_{A}P_{X|S_e,A}$. Hence $I(X;Y|S)$ is concave in $P_{A}P_{X|S_e,A}$. Thus
denoting $R^{\lambda}$ as the value of $I(X;Y|S)$ at joint $P^{\lambda}$, we
have
\bea
\lambda C(\Gamma_1)+(1-\lambda)C(\Gamma_2)\leq R^\lambda\leq C(\lambda
\Gamma_1+(1-\lambda)\Gamma_2).\nonumber
\eea
 \item (e) follows from the fact that $C(\Gamma)$ is non decreasing in $\Gamma$,
which can be argued easily as larger 
$\Gamma$ implies a larger feasible region and hence larger capacity.
\end{itemize} 
We further note the following relations and Markov Chains :
\begin{itemize}
\item $A_i=A_i(M)$ is independent of $S_i$ as state sequence is independent of
message and actions are functions of message.
\item $X_i-(S_{e,i},A_i)-S_{i}$. Refer to Appendix \ref{markov1} for Proof. 
\item $Y_i-(X_i,S_i)-(A_i,S_{e,i})$ follows from the DMC assumption on the
channel which implies the induced joint probability distribution as in
Eq. (\ref{mainjoint}).
\end{itemize}
Hence by using Equations (\ref{mutual1}), (\ref{fano1}) and (\ref{costcap1}),
and
letting
$n\rightarrow\infty$ we have $R\leq C(\Gamma)$.
\end{proof}
\begin{note}[Causal Probing]
 Note that the capacity is the same if we now consider the setting where the
encoder
generates channel input sequences using observed state causally. It is easier to
see that converse
holds without change as in non-causal setting. Achievability remains same
because we are multiplexing
based only on current observed partial state
 information.
\end{note}
\begin{note}[Probing Independent of Messages]
 If action sequence is taken independent of message, \textit{time sharing} is
optimal. This is because when action sequence is
independent of message, the setting is equivalent
 to the case when decoder knows the action. The
capacity in this case is,
\bea
C(\Gamma)&=&\max[I(X;Y|S,A)]\\
&=&\max[p(A=0)I(X;Y|S,A=0)\nonumber\\ &+&p(A=1)I(X;Y|S,A=1)]\\
&=&p(A=0)C(0)+p(A=1)C(1).
\eea
\end{note}
\section{Equivalence between Encoder Probing and Channels with Action-Dependent
States}
\label{equivalence}
In the previous section we motivated the basic problem of characterizing the
capacity when observation of the channel state at the encoder comes at a price.
We had further assumed that the decoder knew the
complete state information. In this section, we point out the
equivalence of
general setting of action dependent channel probing at the encoder with the
setting of channels with action dependent states considered in
\cite{TsachyChannel}. In our generalized setting,
actions are taken in an alphabet $\mathcal{A}$ and
 encoder observes $S_e$ through a DMC $P_{S_e|S,A}$. The setting in
\cite{TsachyChannel} and \cite{Kittichokechai} is as follows. Given a message 
$M$, encoder takes actions $A^n=A^n(M)$, which affect the formation of channel
states. These
states are then used by the encoder causally or non-causally to generate channel
input.\newline
First consider the case when decoder does not
know the
channel states. Now in our setting we are given from
nature $P_S, P_{S_e|S,A},P_{Y|X,S}$, but this is equivalent to
$P_{S_e|A},P_{Y|X,S_e,A}$ since $S^n$ is not available at encoder or decoder
and 
hence can be averaged out. This establishes the
equivalence as depicted in Table I and Fig. \ref{CompareStateEst}. If the
decoder now knows
the channels state $S_d$ through DMC $P_{S_d|S,S_e,A}$ we can replace
$\tilde{Y}$ in Fig. \ref{CompareStateEst} with $(Y,S_d)$
to compute capacity.
\begin{table}
\begin{center}
\caption{Equivalence of setting in \cite{TsachyChannel} to our formulation of
optimal probing at encoder.}
\begin{tabular}{|p{4cm}|p{4cm}|}
\hline
Action Dependent State Channels (\cite{TsachyChannel}) & Optimal Encoder
 State Probing\\\hline
$\tilde{X}$ & $X$\\\hline
$\tilde{A}$ & $A$\\\hline
$\tilde{S}$ & $S_e$\\\hline
$\tilde{Y}$ & $(Y,S_d)$\\\hline
\end{tabular}
\end{center}
\end{table}
\begin{figure}[htbp]
\begin{center}
\scalebox{0.5}{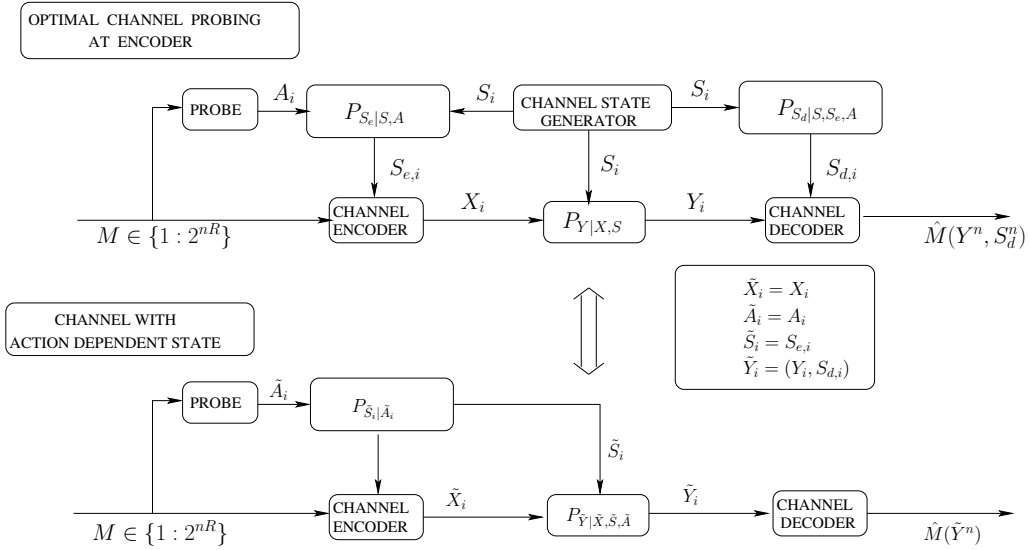}
\caption{Equivalence of our setting of probing the channel state at the encoder
to that of channels with action
dependent states in \cite{TsachyChannel}.}
\label{CompareStateEst}
\end{center}
\end{figure}

Hence using the proven equivalence we invoke and list theorems from
\cite{TsachyChannel}
transformed for our setting.
\begin{theorem}[Equivalent to Theorem 1 in \cite{TsachyChannel}.]
\label{theorem2}
The \textit{\textquoteleft probing capacity\textquoteright}  for optimal channel
state
observation at the encoder which generates channel
inputs using partial state information \textit{non-causally} as in Fig.
\ref{CompareStateEst} with cost constraint $\Gamma$, is given
by,
\bea\label{cap3}
C_{nc}(\Gamma)&=&\max[I(U;Y,S_d)-I(U;S_e|A)]\\
&=&\max[I(A,U;Y,S_d)-I(U;S_e|A)],
\eea
where maximization is over all joint distributions of the form
\bea
P_{A,S,S_e,U,S_d,X,Y}&=&P_AP_SP_{S_e|S,A}P_{U|S_e,A}P_{S_d|S,S_e,A}\nonumber\\
&\times&\1_{\{X=f(U,S_e)\}}P_{Y|X,S},
\eea
for some $P_A,P_{U|S_e,A},f$ such that $E[\Lambda(A)]\leq\Gamma$ and
$\card{\mathcal{U}}\leq\card{\mathcal{A}}\card{\mathcal{S}}\card{\mathcal{S}_e}
\card { \mathcal { S } _d
} \card{\mathcal{X}}
+3$.
\end{theorem}

\begin{theorem}[Equivalent to Theorem 2 in \cite{TsachyChannel}.]
\label{theorem3}
The \textit{\textquoteleft probing capacity\textquoteright} for optimal channel
state
observation at the encoder which generates channel
inputs using partial state information \textit{causally} as in Fig.
\ref{CompareStateEst} with
cost constraint $\Gamma$ is given by,
\bea\label{cap2}
C_{c}(\Gamma)=\max[I(U;Y,S_d)],
\eea
where maximization is over all joint distributions of the form
\bea
P_{U,A,S,S_e,S_d,X,Y}&=&P_U\1_{\{A=g(U)\}}P_SP_{S_e|S,A}P_{S_d|S,S_e,A}
\nonumber\\
&\times& \1_{ \{X=f(U,S_e)\}}P_{Y|X,S},
\eea
for some $P_U,g,f$ such that $E[\Lambda(A)]\leq\Gamma$ and
$\card{\mathcal{U}}\leq\min\{\card{\mathcal{Y}}\card{\mathcal{S}_d},\card{
\mathcal{A}}\card{\mathcal{S}}\card{\mathcal{S}_e}\card{\mathcal{S}_d}\card{
\mathcal{X}}+3\}$
\end{theorem}

\begin{note}
 Note that auxiliary variable $U$ has an increased cardinality as compared to
equivalent setting in \cite{TsachyChannel}. This stems from the following,
\begin{itemize}
 \item Output $Y$ is replaced with $(Y,S_d)$, hence in causal setting we have
$\card{\mathcal{U}}\leq \card{\mathcal{Y}}\card{\mathcal{S}_d}$ following the
arguments in \cite{TsachyChannel}.
\item To preserve $P_{A,S,S_e,S_d,X}$, in both causal and non-causal setting we
have
$\card{\mathcal{U}}\leq\card{\mathcal{A}}\card{\mathcal{S}}\card{\mathcal{S}_e}
\card{\mathcal{S}_d
} \card{\mathcal{X}}
-1$. In causal setting, four more elements are needed, one to preserve
$H(Y,S_d|U)$, one to preserve independence of $S$ with $(A,U)$ and two more
each to preserve markov chains $(S_e,S_d)-(S,A)-U$ and $X-(U,S_e)-(A,S,S_d)$.
In non causal setting, four more elements are needed, one to preserve
$H(S_e|A,U)-H(Y,S_d|U)$, one to preserve independence of $S$ with $A$ and two
more to preserve markov chains, $U-(S_e,A)-(S,S_d)$ and $X-(U,S_e)-(A,S,S_d)$.
\end{itemize}

\end{note}

\par
\textit{Deriving Theorem \ref{theorem1} using Theorems \ref{theorem2} and
\ref{theorem3}}
We would like to derive the capacity results in Theorem \ref{theorem1}
from Theorems  \ref{theorem2} and \ref{theorem3}. We have already pointed out
that capacity of the setting  in Fig. \ref{EncoderStateEstLogicD} is the same
whether encoder encodes using partial information causally or non-causally (call
it $C(\Gamma)=C_c(\Gamma)=C_{nc}(\Gamma)$). (Subscripts \textquoteleft
c\textquoteright and \textquoteleft nc\textquoteright \ stand for capacity for
causal and non-causal encoding of partial state information). We claim to
prove the
result
$C(\Gamma)=C_c(\Gamma)=C_{nc}(\Gamma)$ using Theorems \ref{theorem2}
and \ref{theorem3}.\newline
For non-causal encoding (using Theorem \ref{theorem2})
\bea
C_{nc}(\Gamma)&=&\max[I(A,U;Y,S)-I(U;S_e|A)]\\
&=&\max[I(A,U;Y|S)+I(A,U;S)\nonumber\\
&& -I(U;S_e|A)]\\
&\stackrel{(a)}{=}&\max[H(Y|S)-H(Y|S,A,U,S_e,X)\nonumber\\
&& + I(U;S|A)-I(U;S_e|A)]
\\
&=&\max[H(Y|S)-H(Y|S,A,U,S_e,X)\nonumber\\
&& - H(U|S,A,S_e)+H(U|S_e,A)]\\
&\stackrel{(b)}{=}&\max[H(Y|S)-H(Y|S,X)-H(U|A,S_e)\nonumber\\
&& + H(U|S_e,A)]\\
&=&I(X;Y|S),\label{eq1}
\eea
where 
\begin{itemize}
\item (a) follows from the fact that $S_e=h(S,A)$ and $X=f(U,S_e)$ and
that
$A$ is independent of $S$.
\item (b) follows from the DMC ($P_{Y|X,S}$) assumption and that $U-(S_e,A)-S$
is a Markov Chain.
\end{itemize}

This maximization is over joint distribution
\bea
\nonumber
&&P_{A,S,S_e,U,X,Y}\nonumber\\
&=&P_AP_SP_{S_e|S,A}P_{U|S_e,A}
\nonumber\\
&\times&\1_{\{X=f(U,S_e)\}}P_{Y|X,S}\\
&\stackrel{(c)}{=}&P_A(a)P_S(s)P_{S_e|S,A}(s_e|s,
a)\nonumber\\
&&P_{U|S_e}(u|s_e)\1_{\{x=f(u,s_e)\}}P_{Y|X,S}(y|x,
s)\nonumber\\\\
&=&P_AP_SP_{S_e|S,A}P_{X|S_e}P_{Y|X,S},\nonumber\\
\label{eq2}
\eea
where (c) follows from the fact that knowing $S_e$ implies knowing $A$. Hence we
have from Equations (\ref{eq1}) and (\ref{eq2}).
$C_{nc}(\Gamma)=C(\Gamma)$. \newline
Now for causal encoding (using Theorem \ref{theorem3})
\bea
C_{c}(\Gamma)&=&\max[I(U;Y,S)]\\
&\stackrel{(d)}{=}&\max[I(U;Y|S)]\\
&\stackrel{(e)}{=}&\max[I(A,U;Y|S)]\\
&\stackrel{}{=}&\max[H(Y|S)-H(Y|S,A,U,S_e,X)]\\
&\stackrel{}{=}&\max[H(Y|S)-H(Y|S,X)]\\
&=&I(X;Y|S),\label{eq4}
\eea
where (d) follows from the fact that $U$ and $S$ are independent and (e) follows
from the relation $A=g(U)$.
This maximization is over joint distribution
\bea
P_{U,A,S,S_e,X,Y}&=&P_U\1_{\{A=g(U)\}}P_SP_{S_e|S,A}\nonumber\\&\times&\1_{\{
X=f(U,S_e)\}}P_{Y|X,S}.\label{eq3}
\eea
We will now show that joint distribution of the form in Theorem \ref{theorem1}
is contained in (\ref{eq3}). So the joint distribution in Theorem \ref{theorem1}
\bea
P_{A,S,S_e,X,Y}&=&P_AP_SP_{S_e|S,A}P_{X|S_e,A}P_{Y|X,S}\\
&\stackrel{(f)}{=}&P_SP_{S_e|S,A}P_QP_A\1_{\{X=F(S_e,A,Q)\}}P_{Y|X,S}
\nonumber\\\\
&\stackrel{(g)}{=}&P_U\1_{\{A=g(U)\}}P_SP_{S_e|S,A}\nonumber\\
&\times&\1_{\{X=F(U,S_e)\}}P_{Y|X,S},\label{eq5}
\eea
where (f) follows from the Functional Representation Lemma (\cite{GamalKim}),
$Q$ is independent of $S_e,A$ and (g) follows from defining $U=(A,Q)$. Hence by
Equations (\ref{eq4}) and (\ref{eq5}) we have shown that $C_c(\Gamma)\geq
C(\Gamma)$.
But $C_c(\Gamma)\leq C_{nc}(\Gamma)=C(\Gamma)$. This
completes the claim.

\section{Optimal Probing at Both Encoder and Decoder}
\label{causal}
In earlier sections we considered the framework where only encoder was allowed
to take actions. In this section 
we further generalize the setting where decoder can also take actions based on
 the channel output and then obtain its own 
 partial state information which is used to construct estimate of the
transmitted message. We motivate this general setting in the framework
of communication over slow fading Channels. 

\begin{figure}[htbp]
\begin{center}
\scalebox{0.7}{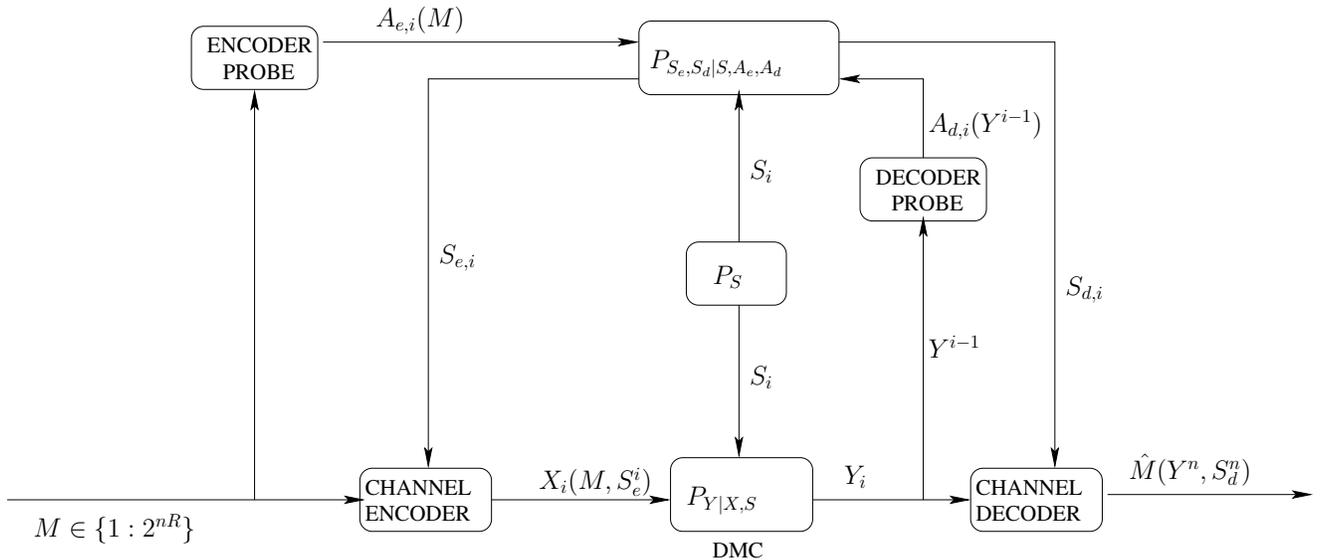}
\caption{Encoder and decoder both take actions to observe partial state
information and use it for encoding and decoding.}
\label{ProbingAction}
\end{center}
\end{figure}
Consider a point to
point communication system where in each time epoch channel state is i.i.d.
$\sim P_S(s_i),$ $s_i\in\mathcal{S}$. In the next epoch
the information of this
present 
state is lost, hence encoder and decoder have to exploit whatever information is
available to them causally to get the best achievable rate. More precisely
consider the setup as depicted in Fig. \ref{ProbingAction} : Message
$M$ is selected uniformly from a uniform distribution on the message set
$\mathcal{M}=\{1,2,\cdots,\card{\mathcal{M}}\}$. Nature generates states
sequence
$S^n\in \mathcal{S}^n$ i.i.d $\sim P_S$, independent of message.
A $(2^{nR},n)$ code consists of :
\begin{itemize}
\item \textit{Probing Logic} : 
\begin{itemize}
 \item \textit{Encoder Probing Logic} $f_{A_{e,i}}:M\rightarrow A_{e,i}\in
\mathcal{A}_e$
 \item \textit{Decoder Probing Logic} $f_{A_{d,i}}:Y^{i-1}\rightarrow A_{d,i}\in
\mathcal{A}_d$, where channel output $Y\in\mathcal{Y}$.
\end{itemize}
Further the encoder and decoder actions are cost constrained, 
\bea
\Lambda(A_e^n,A_d^n)=\frac{1}{n}\sum_{i=1}^{n}\Lambda(A_{e,i},A_{d,i})\leq
\Gamma,
\eea
where $\Lambda(\cdot,\cdot)$ is the cost function while $\Gamma$ is the cost
constraint. Given nature generated state sequence $S^n$, message dependent
encoder action sequence $A_e^n$ and channel output dependent decoder action
sequence $A_d^n$, encoder acquires partial state information $S_e^n\in
\mathcal{S}_e^n$ (which we will call CSIT, i.e. Channel State Information at
Transmitter) and decoder $S_d^n\in
\mathcal{S}_d^n$ (which we will call CSIR, i.e. Channel State Information at
Receiver), through a DMC $P_{S_e,S_d|S,A_e,A_d}$.
\item \textit{Encoding} : $f_{e,i} : (M,S^i_e)\rightarrow X_i\in\mathcal{X}$.
\item \textit{Decoding} : $f_{d} : (Y^n,S_d^n)\rightarrow \hat{M}\in\mathcal
\{1,2,\cdots,\card{\mathcal{M}}\}$. 
\end{itemize}
The joint PMF on $(M,A_e^n,A_d^n,S^n,S^n_e,S_d^n,X^n,Y^n,\hat{M})$ induced by
a given scheme
is
\bea
&&P_{M,A^n,S^n,S^n_e,S^n_d,X^n,Y^n,\hat{M}}(m,a^n,s^n,s^n_e,s_d^n,x^n,y^n,\hat{m
}
)\nonumber\\
&=&\frac{1}{\card{\mathcal{M}}}\prod_{i=1}^{n}
\1_{\{a_{d,i}=f_{A_{d,i}}(y^{i-1}\}}
\1_{\{a_{e,i}=f_{A_{e,i}}(m)\}}P_S(s_i)P_{S_e,S_d|S,A_e,A_d}(s_{e,i},s_{d,i}|s,
a_ {e,i},a_{d,i})\\
&&\times\prod_{i=1}^{n}\1_{\{x_i=f_{e,i}(m,s_{e}^i)\}}P_{Y|X,S}(y_i|x_i,
s_i)\times\1_{\{\hat{m}=f_d(y^n,s_d^n)\}}.
\eea
\subsubsection{Probing Capacity}
\begin{theorem}
\label{theorem4}
The cost constrained \textit{\textquoteleft probing capacity\textquoteright}
for the scenario depicted in Fig.
\ref{ProbingAction} is given by
\bea\label{cap3}
C(\Gamma)=\max[I(U;Y,S_d|A_d)],
\eea
where maximization is over all joint distributions of the form
\bea
&&P_{S,A_d,U,A_e,S_e,X,Y,S_d}(s,a_d,u,a_e,s_e,x,y,s_d)\nonumber\\
&=&P_S(s)P_{A_d}(a_d)P_{U|A_d}(u|a_d)\1_{\{a_e=g(u,a_d)\}}P_{S_e,S_d|S,A_e,A_d}
(s_e,s_d|s ,
a_e,a_d)\1_{\{x=f(u,s_e,a_d)\}}P_{Y|X,S},
\eea
for some $P_{A_d},P_{U|A_d},g,f$ such that $\E[\Lambda(A_e,A_d)]\leq\Gamma$ and
$\card{\mathcal{U}}\leq\min\{\card{\mathcal{Y}}\card{\mathcal{S}_d}\card{
\mathcal{A}_d} ,  \card{
\mathcal{S}}\card{
\mathcal{A}_d}\card {
\mathcal{A}_e}\card{\mathcal{S}_e}\card{\mathcal{S}_d}\card{\mathcal{X}}+4\}$
\end{theorem}
\begin{proof}
\newline
\textit{Achievability} : Fix $P_{A_d},P_{U|A_d},g,f$ which achieve
$C(\frac{\Gamma}{1+\epsilon})$. Encoder and decoder decide on a sequence
$A_d^n$, i.i.d $\sim P_{A_d}$. By similar arguments as in achievability of
previous theorems using typical average lemma, constraints are satisfied. Now
using Theorem \ref{theorem2} if $A_{d,i}=a$ $\forall i$, error free
communication is achieved if $R<I(U;Y,S_d|A_d=a)$. Hence since encoder and
decoder both know $A_d^n$, we achieve $R<I(U;Y,S_d|A_d)$.
\newline
\textit{Converse} :
Suppose rate $R$ is achievable. Now consider a sequence of $(2^{nR},n)$ codes
for which we have $P_e^n\stackrel{n\rightarrow\infty}{\longrightarrow}$.
Consider
\bea
nR&=&H(M)\\
&=&I(M;Y^n,S^n_d)+H(M|Y^n,S_d^n).\label{mutual2}
\eea
By Fano's Inequality (\cite{CoverThomas})
\bea
\label{fano2}H(M|Y^n,S_d^n)\leq 1+P^n_eR\leq n\epsilon_n,
\eea
where $\epsilon_n\stackrel{n\rightarrow\infty}{\longrightarrow}0$. Now Consider
\bea
I(M;Y^n,S_d^n)&=&H(Y^n,S_d^n)-H(Y^n,S_d^n|M)\\
\label{conveq1}&\stackrel{(a)}{=}&\sum_{i=1}^{n}H(Y_i,S_{d,i}|Y^{i-1},S_d^{i-1},
A_d^i)-\sum_{i=1}
^{n}H(Y_i,S_{d,i}|Y^{i-1},
S^{i-1}_d,M,A_d^i,A_e^n)\\
&\stackrel{}{\leq}&\sum_{i=1}^{n}H(Y_i,S_{d,i}|A_{d,i})-\sum_{i=1}^{n}H(Y_i,S_{d
,
i}|Y^{i-1},
S^{i-1}_d,S_e^{i-1},M,A_d^i,,A_e^n)\\
&\stackrel{(b)}{\leq}&\sum_{i=1}^{n}H(Y_i,S_{d,i}|A_{d,i})-\sum_{i=1}^{n}H(Y_i,
S_{d,
i}|U_i,A_{d,i})\\
&=&\sum_{i=1}^{n} I(U_i;Y_i,S_{d,i}|A_{d,i})\\
&\stackrel{}{\leq}&\sum_{i=1}^{n}C(\E[\Lambda(A_{e,i},A_{d,i})])\\
&\stackrel{(c)}{\leq}&nC(E[\Lambda(A_e^n,A_d^n)])\\
&\stackrel{(d)}{\leq}&n C(\Gamma),\label{costcap2}
\eea
\begin{itemize}
 \item (a) follows from the fact that $A_{d,i}=A_{d,i}(Y^{i-1})$ and
$A_e^n=A_e^n(M)$.
 \item (b) follows by defining
$U_i=(M,Y^{i-1},S_d^{i-1},S_e^{i-1},A_d^{i-1},A_e^n)$.
\item (c) follows from the fact that $C(\Gamma)$ is concave in $\Gamma$. This
is proved in Appendix \ref{concavification}.
 \item (d) follows from the fact that $C(\Gamma)$ is non decreasing in $\Gamma$,
which can be argued easily as larger 
$\Gamma$ implies a larger feasible region and hence larger capacity.
\end{itemize}
We note the following relations,
\begin{itemize}
\item $A_{d,i}=A_{d,i}(Y^{i-1})$ is independent of $S_i$, it follows from proof
of markov chain MC1 in Appendix \ref{markov1}.
\item We have the Markov Chains, 
\begin{itemize}
\item $U_i-A_{d,i}-S_i$.
\item $A_{e,i}-(U_i,A_{d,i})-S_i$.
\item $(S_{e,i},S_{d,i})-(S_i,A_{e,i},A_{d,i})-U_i$.
\item $X_i-(U_i,S_{e,i},A_{d,i})-(A_{e,i},S_i,S_{d,i})$.
\item $Y_{i}-(X_i,S_{i})-(U_i,A_{d,i},A_{e,i},S_{e,i},S_{d,i})$.
\end{itemize}
These are proved in Appendix \ref{markov2}. 
\item As $U_i$ contains $A^n_e$, maximization is unaffected if we replace
$P_{A_e|U,A_d}$ with $\1_{\{A_e=g(U,A_d)\}}$. Since $I(U;Y,S_d|A_d)$ is
convex in $P_{Y,S_d|U,A_d}$, this implies convexity in
$P_{X|U,S_e,A_d}$. hence again maximum would
be unaffected if general $P_{X|U,S_e,A_d}$ is replaced with
 $X=f(U,S_e,A_d)$.
\item \textit{Cardinality Bounds on U} That set $\mathcal{U}$
needs no more than $\card{\mathcal{Y}}\card{\mathcal{S}_d}\card{\mathcal{A}_d}$
follows from arguments in \cite{Salehi}. Also $\mathcal{U}$ needs
$\card{\mathcal{S}}\card{\mathcal{S}_e}\card{\mathcal{A}_e}
\card{\mathcal{A}_d} \card{\mathcal{S}_d} \card{\mathcal{X}}
-1$ to preserve $P_{S,A_e,A_d,S_e,S_d,X}$ (which preserves
$H(Y_d,S_d|A_d))$, one
element to preserve $H(Y,S_d|A_d,U)$, one element to preserve independence of
$S$ and $A_d$ and three more to preserve the markov
chains, $(U,A_e)-A_d-S$, $(S_e,S_d)-(S,A_e,A_d)-U$ and
$X-(U,S_e,A_d)-(S,S_d,A_e)$.
\end{itemize}
The proof is then completed by using Eq. (\ref{mutual2}), (\ref{fano2}) and
(\ref{costcap2}).
\end{proof}
\begin{note}
 We can consider a more general setting where encoder and decoder feedback
logic depend upon the respective past state observations, i.e., encoder takes
actions, $A_{e,i}(M,S_e^{i-1})$, while decoder takes actions,
$A_{d,i}(Y^{i-1},S_d^{i-1})$. While the achievability remains unchanged as in
Theorem \ref{theorem4}, it is easy to see the converse also hold with
$U_i=(M,Y^{i-1},S_d^{i-1},S_e^{i-1},A_d^{i-1},A_e^i)$.
\end{note}

\begin{note}[Computer Memory with Defects : Non-causal Probing at both Encoder
and Decoder] : Consider a computer memory with defects, as in what the encoder
writes,
$X$ and what the decoder reads, $Y$ are related to each other through a
discrete memoryless channel, $P_{Y|X,S}$, where state $S$ models defects. If
there are no cost constraints to acquire the information about defects, encoder
and decoder are better-off by coding and decoding using this entire state
sequence $S^n$ as it is available before writing and reading on the memory. Note
that we assume neither the writing nor the reading operation changes the state.
However when acquisition of this state information by the encoder as well as
the decoder is cost constrained, encoder can take actions, $A_{e,i}(M)$ to get
partial state information $S_e^n$ and then write $X_i(M,S^n_e)$ while decoder
can wait for entire memory to be written and then take actions, $A_{d,i}(Y^n)$.
It will then obtain its side information $S_d^n$. Hence the setup remains
similar as depicted in Fig. \ref{ProbingAction}, the only difference from
the setup in Section \ref{causal} is that encoder now uses the partial
state information, CSIT,
non-causally to generate input symbols, i.e. $f_e : (M,S^n_e)\rightarrow
X_i\in\mathcal{X}$, while decoder takes action based on entire channel output
sequence, i.e., $f_{A_d}:Y^{n}\rightarrow A_{d,i}\in
\mathcal{A}_d$. Also in order to avoid issues of instantaneous dependency, we
must have,
\bea
P_{S_e,S_d|S,A_e,A_d}=P_{S_e|S,A_e}\times P_{S_d|S,S_e,A_e,A_d}
\eea
\newline 
\textit{ Equivalence to Relay Problem } \newline
The above problem is in general a hard one. Consider a special case where
$A_e$ is binary, with cost function $\Lambda(A_e,A_d)=\Lambda(A_e)$. For this
case, the zero cost and unit
cost corner cases are themselves open with only bounds. When cost is unity,
this is the case of relay channel with states and infinite lookahead with
states known non causally to the encoder. For the standard relay channel (no
infinite lookahead) with states known to encoder, Zaidi and Vanderdorpe in
\cite{ZaidiVandendorpe} lower bound the capacity. For zero cost
the system is a special case of 'Relay Channel with Infinite Lookahead'. We
conclude by showing the equivalence of
this 
problem at zero cost to that of \textit{Relay with Infinite Lookahead}, as
depicted in in Fig. \ref{CompareRelay} and Table II.
\begin{figure}[htbp]
\begin{center}
\scalebox{0.5}{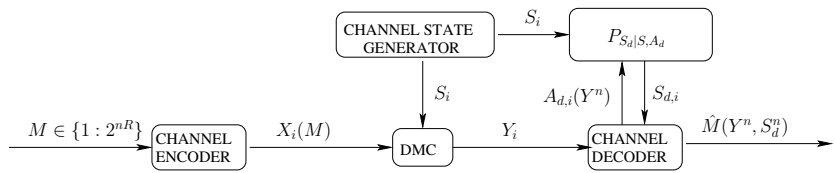}
\caption{Decoder takes actions dependent upon the entire observed channel output
sequence
and uses the actions to aquire partial channel state information. Encoder has
no knowledge of channel states.}
\label{DecoderStateEstimation}
\end{center}
\end{figure}
\begin{figure}[htbp]
\begin{center}
\scalebox{0.5}{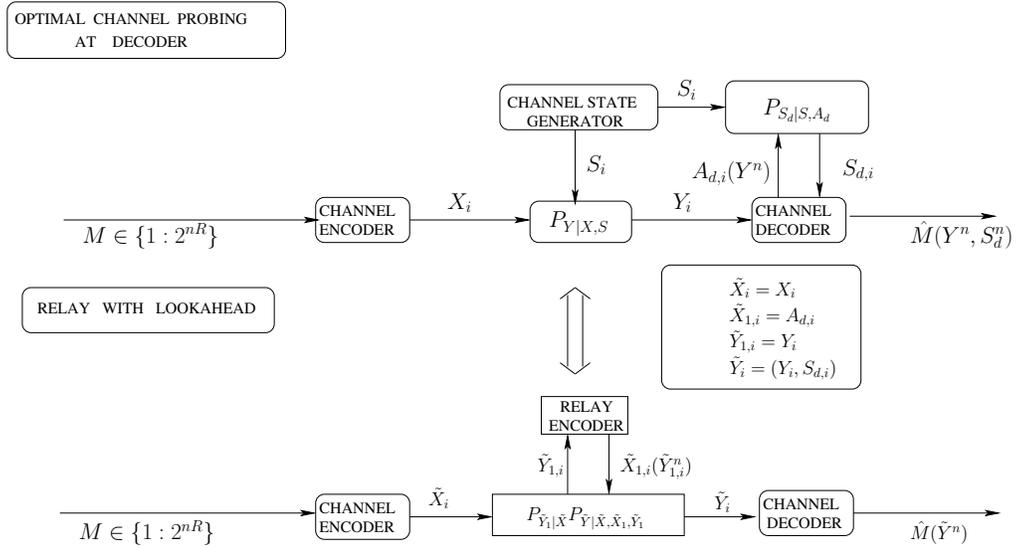}
\caption{Equivalence of setting in Fig. \ref{DecoderStateEstimation} with
Relay with Infinite Lookahead.}
\label{CompareRelay}
\end{center}
\end{figure}

\begin{table}
\begin{center}
\caption{Equivalence of setting in Fig. \ref{DecoderStateEstimation} with
Relay with Infinite Lookahead \cite{GamalKim}.}
\begin{tabular}{|p{4cm}|p{4cm}|}
\hline
Relay with Infinite Lookahead (\cite{GamalKim}) & Decoder Probing in
Fig. \ref{DecoderStateEstimation}\\\hline
$\tilde{X}$ & $X$\\\hline
$\tilde{X_1}$ & $A$\\\hline
$\tilde{Y_1}$ & $Y$\\\hline
$\tilde{Y}$ & $(Y,S_d)$\\\hline
\end{tabular}
\end{center}
\end{table}
\end{note}

\section{Numerical Examples}
\label{example}
\subsection{Discrete Channels}

\subsubsection{[Non-causal Probing] : To Observe or Not to Observe Channel State
at Encoder, Decoder observes complete channel state.}
\begin{figure}[htbp]
\begin{center}
\scalebox{0.5}{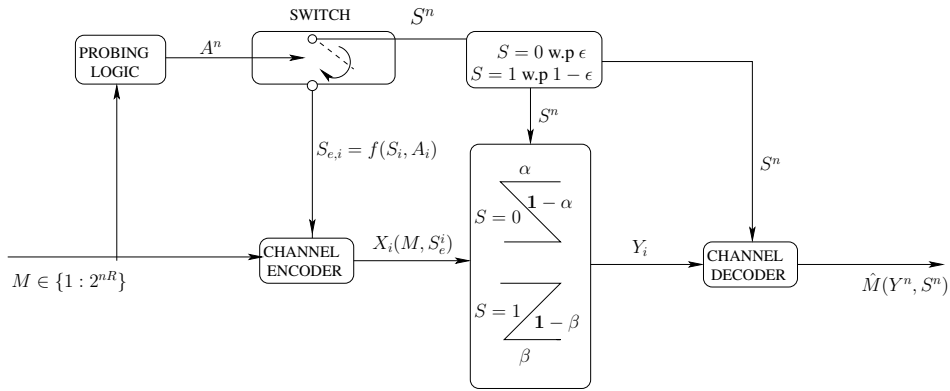}
\caption{Example 1}
\label{example1}
\end{center}
\end{figure}
\begin{example}[Binary States, $S(\alpha)$ channel and $Z(\beta)$] 
Consider the communication system shown in Fig. \ref{example1} with binary input
and output. Decoder knows the state completely. Actions are binary which
correspond \textit{to observe or not to observe} state at encoder. Also the
cost function, $\Lambda(a)=a$, for actions, $a\in\{0,1\}$. We compute the
capacity using Theorem \ref{theorem1}. $S_e\in\{\ast,0,1\}$ and
$\alpha=\beta=\epsilon=0.5$. We assume the
following
\bea
P(X=0|S_e=\ast)=p_1,\\
P(X=0|S_e=0)=p_2,\\
P(X=0|S_e=1)=p_3.
\eea
As $C(\Gamma)$ is non decreasing in $\Gamma$. $P(A=1)=\Gamma$. We obtain for
$\Gamma\in[0,1],$
\bea
\label{opteq}
&&C(\Gamma)\nonumber\\
&=&\max_{p_1,p_2,p_3\in[0,1]}
[\epsilon h_2\left(\alpha((1-\Gamma)p_1+\Gamma
p_2)\right)\nonumber\\
&-&\epsilon((1-\Gamma)p_1+\Gamma p_2)h_2(\alpha)\nonumber\\
&+&(1-\epsilon)h_2\left(\beta((1-\Gamma)(1-p_1)+\Gamma
(1-p_3))\right)\nonumber\\
&-&(1-\epsilon)((1-\Gamma)(1-p_1)+\Gamma (1-p_3))h_2(\beta)].
\eea
We compute the above expression numerically (Fig. \ref{examplefig1}).
\begin{figure}[htbp]
\begin{center}
\includegraphics[scale=0.5]{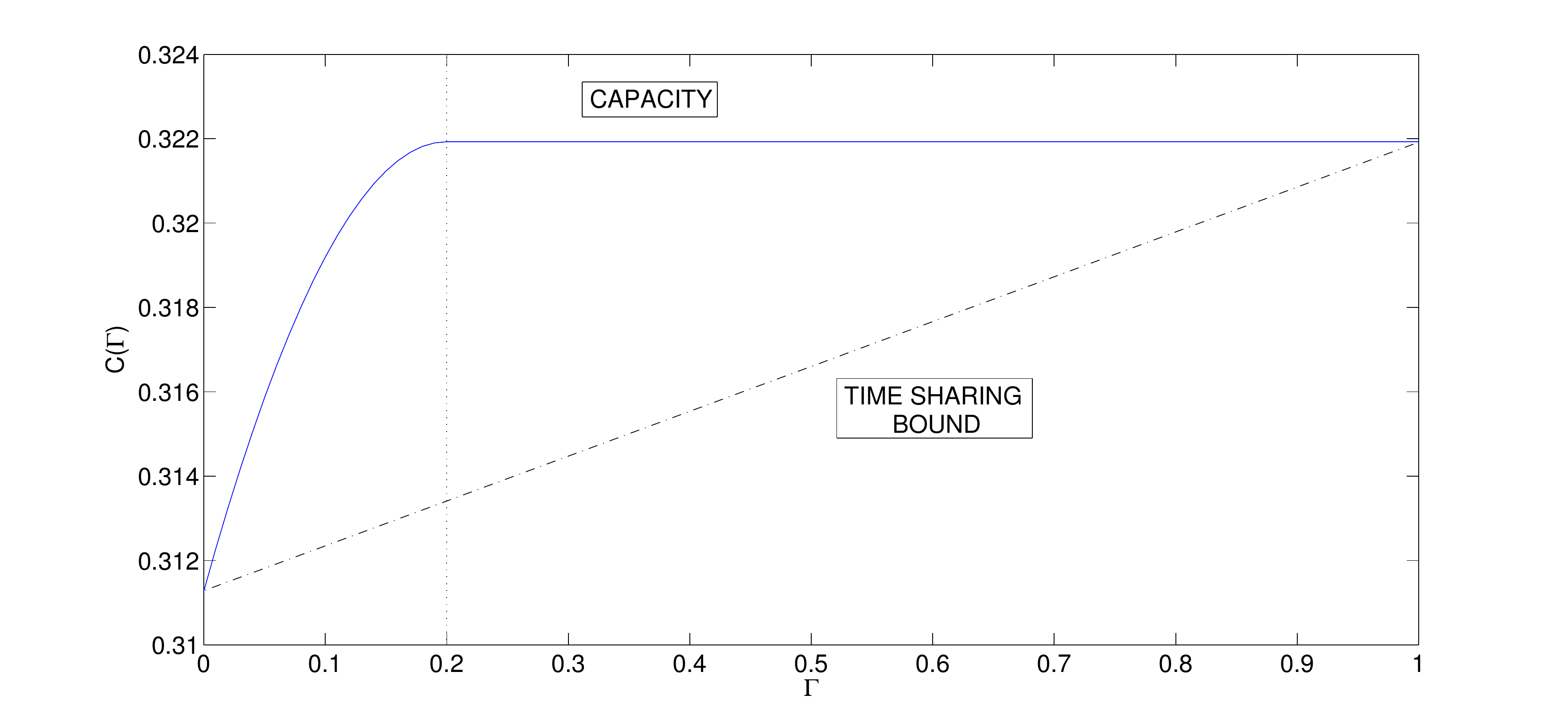}
\caption{Cost-capacity trade off for Example 1. Time sharing is strictly
sub-optimal.}
\label{examplefig1}
\end{center}
\end{figure}
\end{example}

\begin{note}[Cut-off point $\approx 0.2$ in Fig. \ref{examplefig1}]
 An observation from this example which is really surprising is that in order to
achieve the maximum capacity (which is at 
$\Gamma=1)$ one needs to only observe a fraction of states $\approx 0.2$. This
threshold however can also be theoretically derived. Essentially we find out
the 
range of $\Gamma\in[0,1]$ for which the capacity achieving joint distribution in
$C(\Gamma)$ induces exactly the same marginals, 
$P_{X|S}$ as when the cost is unity. Let $p^{\ast}_1$, $p^{\ast}_2$ and
$p^{\ast}_3$ be optimal distributions for cost $\Gamma$ as in Eq. \ref{opteq}.
The marginals are equal to
\bea
P(X=0|S=0)=(1-\Gamma)p^{\ast}_1+\Gamma p^{\ast}_2\\
P(X=0|S=1)=(1-\Gamma)p^{\ast}_1+\Gamma p^{\ast}_3.
\eea
 For $\Gamma=1$, we can easily compute 
$P(X=0|S=0)=0.4$ and $P(X=0|S=1)=0.6$. Therefore for marginals to be same,
\bea
(1-\Gamma)p^{\ast}_1+\Gamma p^{\ast}_2=0.4\\
(1-\Gamma)p^{\ast}_1+\Gamma p^{\ast}_3=0.6,
\eea
or
\bea
\Gamma (p^{\ast}_3-p^{\ast}_2)=0.2.
\eea
Since $p^{\ast}_2,p^{\ast}_3\in[0,1]$, it is easy to see that if the cost
$\Gamma\stackrel{>}{\sim}0.2$, we can find $(p^{\ast}_1,p^{\ast}_2,p^{\ast}_3)$
such that $C(\Gamma)=C(1)$. At $\Gamma=0.2$, optimal scheme is 
$X=S_e\oplus 1\mbox{ if }S_e\neq\ast$, and $Bern(0.5)$ otherwise.
\end{note}

\subsubsection{[Causal Probing] : To Observe or Not to Observe Channel State
at Encoder , with no channel state at the Decoder.}
\begin{figure}[htbp]
\begin{center}
\scalebox{0.5}{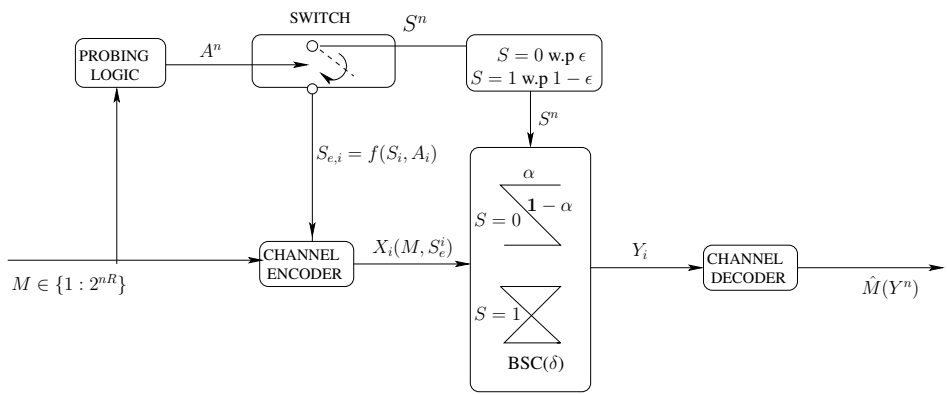}
\caption{Example 2}
\label{example2}
\end{center}
\end{figure}
\begin{example}[Binary States, $S(\alpha)$ channel and $BSC(\delta)$]
Consider the communication system shown in Fig. \ref{example2} with binary input
and output with $\epsilon=0.5$, $\alpha=0.1$ and $\delta=0.3$. Here states are
not known to the decoder and encoder uses partial
state information causally to generate channel input symbols. Actions are
binary with cost, $\Lambda(a)=a$. $A=1$
corresponds to an observation of the channel state while $A=0$ to a lack of an
observation. The
evaluation of
capacity expression involves an auxiliary random variable. We compute its
lower bound on capacity numerically using Theorem as shown in Fig.
\ref{examplefig2}. Here also clearly time sharing is not optimal.
\begin{note}
 Note the interesting phenomenon in this example too (as in Example 1), where we
just need to observe roughly a fraction of state $\sim 0.5$ to obtain the
capacity at unit cost. This can be reasoned in a similar manner as reasoned for
Example 1.
\end{note}
\begin{figure}[htbp]
\begin{center}
\includegraphics[scale=0.5]{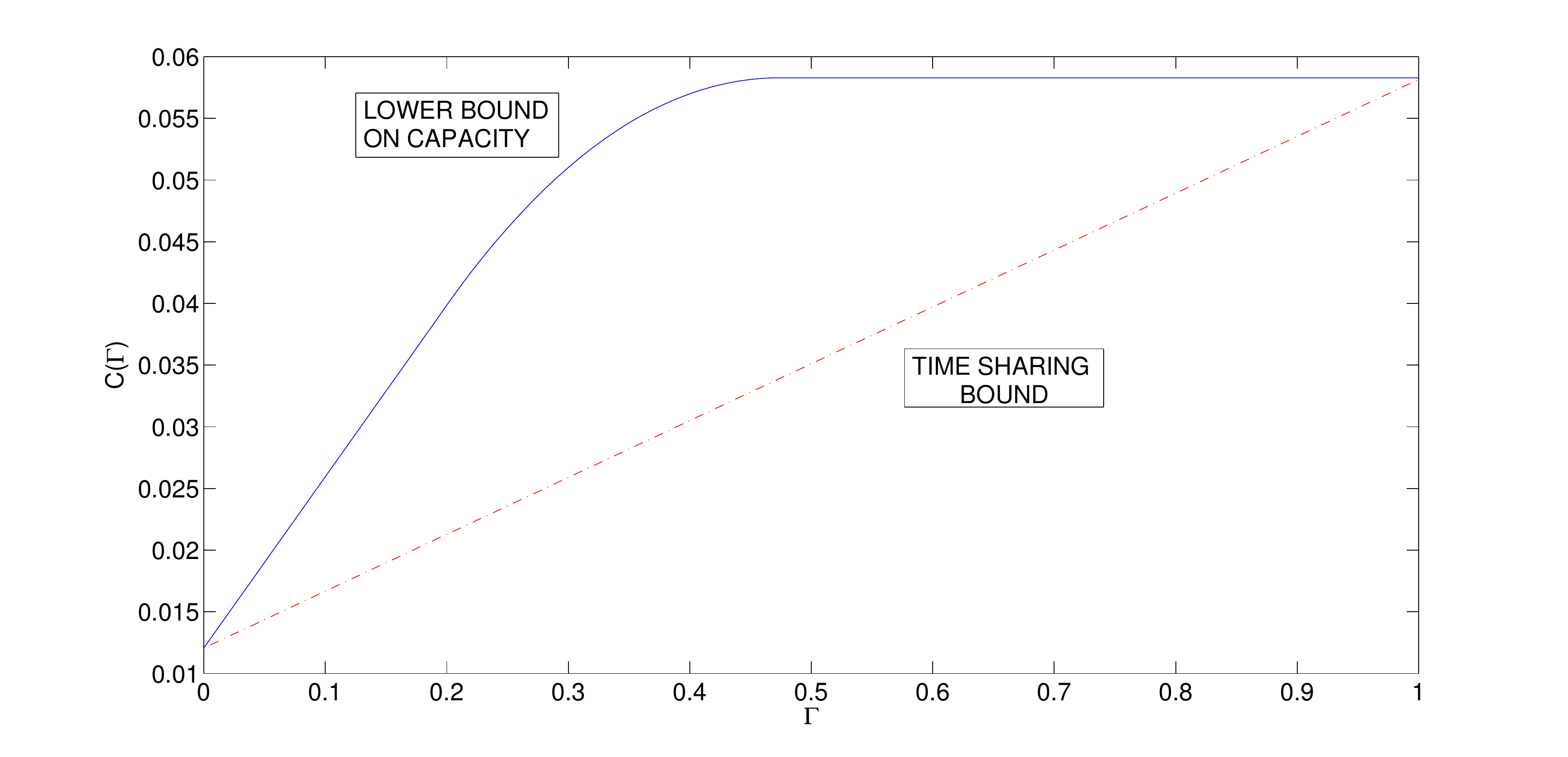}
\caption{Cost-capacity trade off for Example 2. The dotted straight line is
obtained by time sharing between zero cost and unit cost capacity (Scheme 1).
Time sharing between a scheme for which $A=g(U)=U$ in
Theorem \ref{theorem2} (call it Scheme 2) and Scheme 1 gives a lower bound on
the capacity
indicated by solid line. It is evident that naive Scheme 1
(time sharing scheme between extreme capacities at zero and unit cost) is
strictly sub-optimal.}
\label{examplefig2}
\end{center}
\end{figure}
\end{example}

\begin{example}[Binary States, Multiplier Channel with Power Constraints.]
Consider a multiplier channel with binary inputs, outputs and states, $Y=S\cdot
X$ where  
$S\sim Bern(0.5)$. Again note that
actions are binary with $\Lambda(a)=a$ and $A=1$
corresponds to an observation of the channel state while $A=0$ to a lack of an
observation. Let 
\bea
p_{\ast}&=&p(x=1|s_e=\ast)\\
p_0&=&p(x=1|s_e=0)\\
p_1&=&p(x=1|s_e=1).
\eea
We see that capacity under the power constraint,
\bea
p(x=1)\le P_0\in[0,1],
\eea
is 
\bea
C(\Gamma,P_0)&=&\max \frac{1}{2}h_2\left[(1-\Gamma)p_\ast+\Gamma p_1\right]\\
&&\mbox{subject to}\nonumber\\
&& (1-\Gamma)p_\ast+\frac{\Gamma}{2} (p_0+p_1)=P_0
\eea
For $P_0=0.25$, we have
\bea
C(\Gamma,P_0=0.25)&=&0.5h_2\left[\frac{1+2\Gamma}{4}\right]\mbox{ if
}\Gamma\in[0,0.5]\\
C(\Gamma,P_0=0.25)&=&0.5 \mbox{ if }\Gamma\in[0.5,1].
\eea
The plot for $P_0=0.25$ is shown in Fig. \ref{examplefig3}.
\begin{figure}[htbp]
\begin{center}
\includegraphics[scale=0.35]{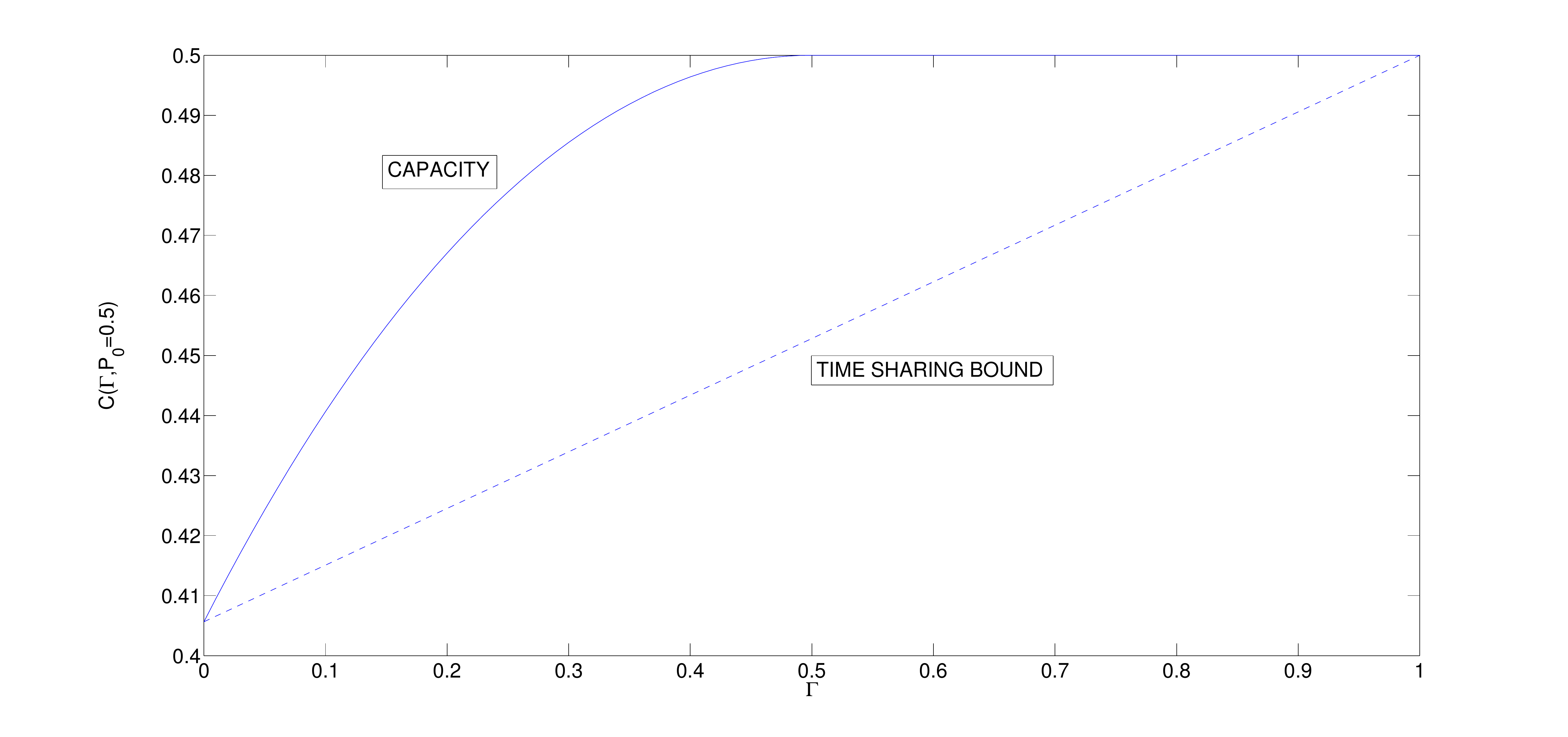}
\caption{Cost-capacity trade off for Example 3 for $P_0=0.25$. The dotted
straight
line is obtained by time sharing between zero cost and unit cost capacity. 
}\label{examplefig3}
\end{center}
\end{figure}
\end{example}

\subsection{Continuous Channels}
\subsubsection{\textquoteleft Learning\textquoteright\ to Write on a Dirty Paper
: }Using standard arguments,
it can be shown that the capacity results carry over to the case of continuous
channels with power 
constraints on input symbols. Let us recall the setting in \textit Dirty Paper
Coding\textit. Costa in 
\cite{CostaDPC} considered the communication system as in Fig.
\ref{DPC}.
\begin{figure}[htbp]
\begin{center}
\scalebox{0.5}{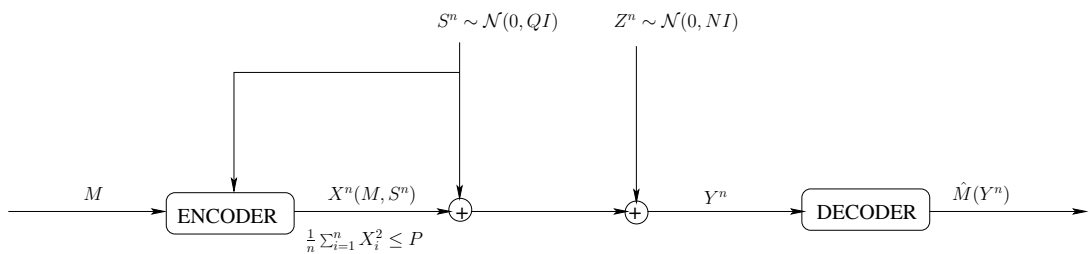}
\caption{Dirty Paper Coding as in \cite{CostaDPC}}
\label{DPC}
\end{center}
\end{figure}
The output of the channel is given as $Y^n=X(M,S^n)+S^n+Z^n$, where 
\begin{itemize}
 \item Channel state or Interference $S^n$ is i.i.d. $S^n\sim\mathcal{N}(0,QI)$
independent of i.i.d. 
noise, $Z^n\sim\mathcal{N}(0,NI)$.
 \item Channel state or interference is known to the encoder non-causally.
Encoder hence generates 
channel inputs $X^n(M,S^n)$ which are cost constrained, i.e.,
$\frac{1}{n}\sum_{i=1}^n X_i^2\le P$.
 \item Decoder has no knowledge of channel state or interference. 
\end{itemize}
It was shown that the capacity of this channel is
$C(P/N)=\frac{1}{2}\log_2(1+P/N)$ which is equal 
to the capacity of a standard gaussian channel with signal to noise ratio $P/N$.
This is strictly larger
 than the capacity when $S^n$ is unknown to both encoder and decoder, i.e.,
 $\frac{1}{2}\log_2(1+P/(Q+N))$. 
\par
We now consider the setting as in Fig. \ref{ADPC}.
\begin{figure}[htbp] 
\begin{center}
\scalebox{0.5}{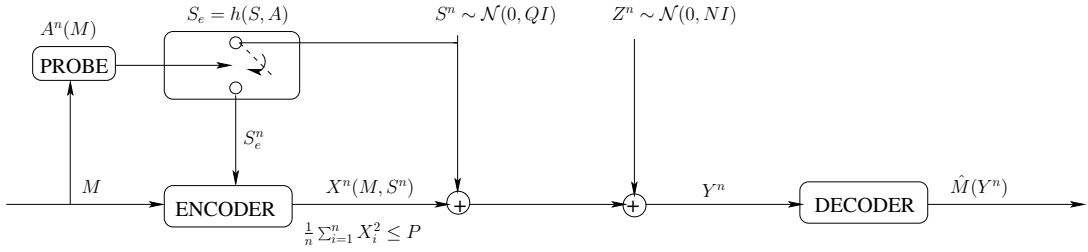}
\caption{Learning to write on a Dirty Paper.}
\label{ADPC}
\end{center}
\end{figure}
 While in 
\textit{Writing on Dirty Paper}, 
it was assumed that interference or channel state was completely available, but
this might not be
true in real systems one might have to pay a price to acquire this information.
Hence in contrast 
to writing on a paper where intensity and positions of all dirt spots are known,
we have to take action
to learn where the paper is most dirty, hence the name \textit{Learning to Write
on a Dirty Paper}. 
Actions are binary, with cost function, $\Lambda(a)=a$. Here also $A=1$
corresponds to an observation of the channel state while $A=0$ to a lack of an
observation. Also, 
\bea
S_e=h(S,A)&=&\ast\mbox{ if }A=0\\
S_e=h(S,A)&=&S\mbox{ if }A=1, 
\eea
where $\ast$ stands for erasure or no information.
\par
Invoking Theorem  \ref{theorem3}, we have the capacity,
\bea
C(\Gamma,P)=\max[I(U;Y)-I(U;S_e|A)]=\max[I(A,U;Y)-I(U;S_e|A)],
\eea
where maximization is over joint distribution,
\bea
f_{A,U,S,S_e,X,Y}=P_A f_S\1_{\{S_e=h(S,A)\}}\1_{\{X=f(U,S_e)\}}f_{Y|X,S}
\eea
such that, $p(A=1)\le\Gamma$ and $\E[X^2]\le P$. We give a lower bound on this
capacity by considering 
a simple \textit{power splitting} achievable scheme. Let us assume
$X|(A=0)\sim\mathcal{N}(0,P_1)$ and 
$X|(A=1)\sim\mathcal{N}(0,P_2)$. Clearly $C(\Gamma,P)$ is maximized when
$p(A=1)=\Gamma$. Therefore we have from power constraints,
\bea
(1-\Gamma)P_1+\Gamma P_2\le P.
\eea
Further we assume, given action $A$, channel input $X$ is independent of
$U,S,Z$. Let
\bea
U|(A=0)&=&X|(A=0)\\
U|(A=1)&=&X|(A=1)+\alpha(P_2) S,
\eea
where $\alpha(P_2)=P_2/(P_2+1)$. Since $Y=X+S+Z$, we have,
\bea
Y|A=0 \sim g_{0}&=&\mathcal{N}(0,P_1+Q+N)\\
Y|A=1 \sim g_{1}&=&\mathcal{N}(0,P_2+Q+N)\\
Y \sim g&=&(1-\Gamma)\mathcal{N}(0,P_1+Q+N)+\Gamma\mathcal{N}(0,P_2+Q+N).
\eea
Considering this distribution gives the following lower bound on capacity, 
\bea
C_{lower}&=&\max_{P_1,P_2}[I(A,U;Y)-I(U;S_e|A)]\\
&=&\max_{P_1,P_2}[I(A;Y)+I(U;Y|A)-I(U;S_e|A)]\\
&\stackrel{(a)}{=}&\max_{P_1,P_2}[h(g)-(1-\Gamma)h(g_0)-\Gamma h(g_1)
+(1-\Gamma)I(X;Y|A=0)+\Gamma(I(U;Y|A=1)-I(U;S|A=1))]\nonumber\\
&\stackrel{(b)}{=}&\max_{P_1,P_2}[h(g)-(1-\Gamma)h(g_0)-\Gamma h(g_1)
+(1-\Gamma)C(P_1/(Q+N))+\Gamma C(P_2/N)],
\eea
where 
\begin{itemize}
 \item (a) follows from the fact that $S_e$ is just erasure for $A=0$, while for
$A=1$ is equal to
 $S$. $h(g)$ denotes the differential entropy of a continuous random variable
with distribution $g$.
 \item (b) follows from the fact that when $A=0$,
\bea
I(X;Y|A=0)&=&h(Y|A=0)-h(Y|X,A=0)\\
&=&h(\mathcal{N}(0,P_1+Q+N))-h(\mathcal{N}(0,Q+N)\\
&=&\frac{1}{2}\log_2(1+P/(Q+N))=C(P/(Q+N)),
\eea
while for $A=1$ following the similar steps as in \cite{CostaDPC}[Eq. 3,4,5,6,7]
we obtain,
\bea
I(U;Y|A=1)-I(U;S|A=1)=\frac{1}{2}\log_2(1+P/N)=C(P/N).
\eea
Fig. \ref{exampleDPC} shows the plot of $C_{lower}$ with $\Gamma$ for $P=Q=N=1$,
which indeed performs better than
naive time sharing between $C(P/N)$ and $C(P/(Q+N))$.
\begin{figure}[htbp]
\begin{center}
\includegraphics[scale=0.35]{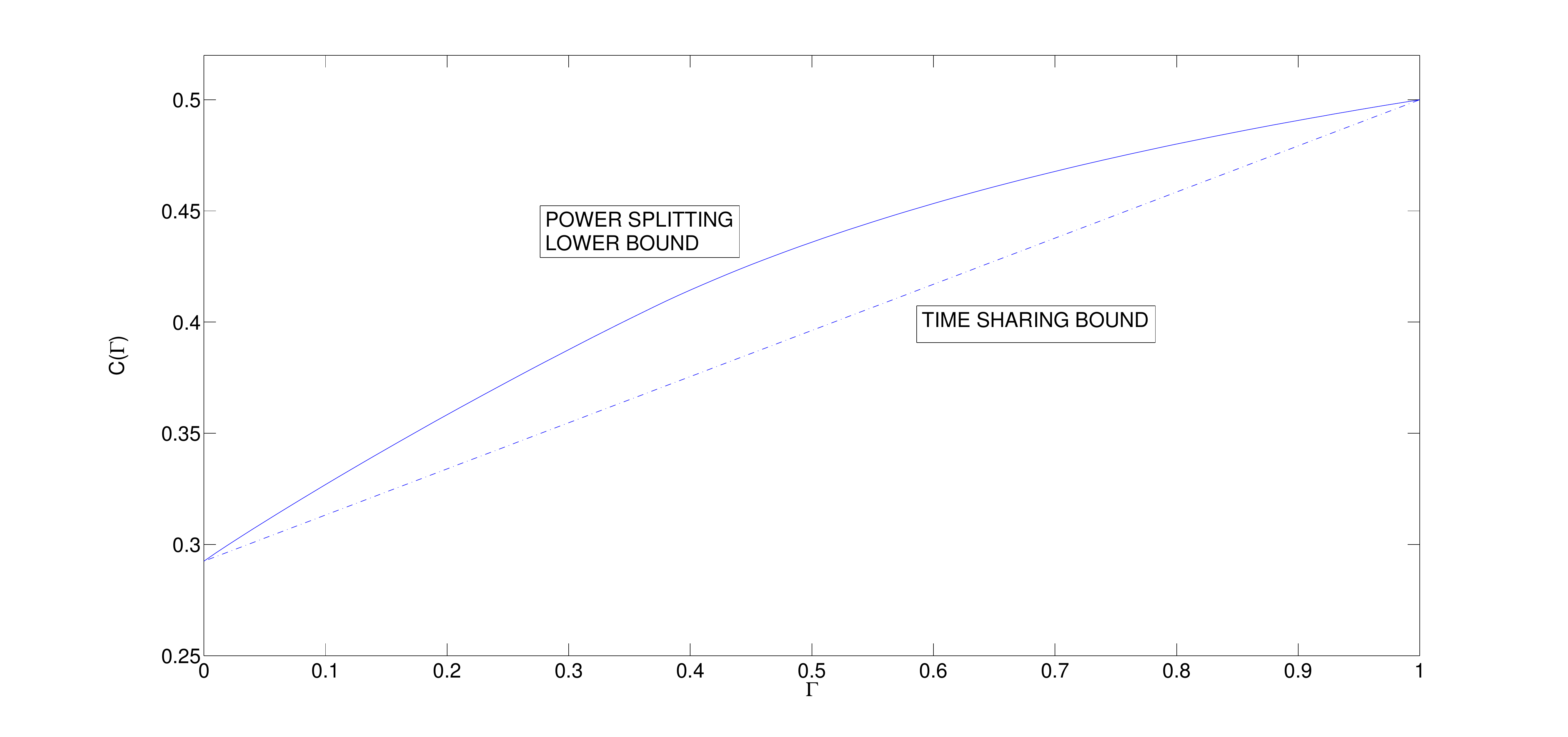}
\caption{Power Splitting lower bound on capacity for \textit{Learning to Write
on Dirty Paper}
 in Fig. \ref{ADPC}.}
\label{exampleDPC}
\end{center}
\end{figure}
\end{itemize}

\subsubsection{Fading Channels with Power Control}
We revisit the setting of fading channels with encoder and decoder state
information as in \cite{GoldsmithVaraiya},
 but now the encoder takes actions to acquire the channel state from receiver
state estimation,
 while decoder already knows
the channel state. This is depicted in Fig. \ref{fading}. Here $g[i]$ denotes
the i.i.d. channel states 
which take value in a finite state, $\mathcal{S}=\{g_1,g_2\}$ with equal
probability. $n[i]$ is i.i.d.
 gaussian noise $\sim \mathcal{N}(0,N/2)$. Bandwidth for communication is $B$.
$\gamma_1=\frac{Pg_1}{NB}$ and $\gamma_2=\frac{Pg_2}{NB}$ are signal to noise
ratios, 
such that $\gamma_1<\frac{\gamma_2}{1+2\gamma_2}$. Actions are binary which
correspond \textit{to observe or not to observe} state at encoder with
cost 
functions $\Lambda(a)=a$ and cost constraint $\Gamma$. $f$ is defined as in
Theorem 
\ref{theorem1}. From results in
\cite{GoldsmithVaraiya}, we know that,
\begin{itemize}
 \item Capacity when only decoder knows the state information
\bea
C(0)=\frac{B}{2}\log_2(1+\gamma_1)+\frac{B}{2}\log_2(1+\gamma_2).
\eea
\item Capacity when encoder also knows the channel state (possibly through a
noiseless feedback 
from decoder) in addition to decoder,
\bea
C(1)=\frac{B}{2}\log_2(1+2\gamma_2).
\eea
\end{itemize}
The above capacities form the extreme cases of zero and unit cost respectively 
for the communication system in Fig. \ref{fading}. Using Theorem \ref{theorem1},
we have the capacity 
for the communication system in Fig. \ref{fading} with bandwidth $B$ as \bea
C=\max_{P_A f_{X|S_e}}2BI(X;Y|S)=\max_{P_A,
f_{X|S_e}}2B[h(Y|S)-h(\mathcal{N}(0,NB))],
\eea
such that $\E[\Gamma(A)]\le\Gamma$ and $\E[X^2]\le P$. Clearly maximum is
attained for 
$p(A=1)=\Gamma$. To obtain a lower bound we assume the following,
\bea
X|(S_e=\ast)&\sim&\mathcal{N}(0,P_{\ast})\\
X|(S_e=g_1)&\sim&\mathcal{N}(0,P_{1})\\
X|(S_e=g_2)&\sim&\mathcal{N}(0,P_{2}).
\eea
This implies,
\bea
Y|(S=g_1)&\sim&(1-\Gamma)\mathcal{N}(0,NB+P_{\ast}g_1)+\Gamma\mathcal{N}(0,
NB+P_1 g_1)\\
Y|(S=g_1)&\sim&(1-\Gamma)\mathcal{N}(0,NB+P_{\ast}g_2)+\Gamma\mathcal{N}(0,
NB+P_2 g_2),
\eea
with power constraints,
\bea
\E[X^2]=(1-\Gamma)P_{\ast}+\frac{\Gamma}{2}(P_1+P_2)\le P.
\eea
Hence a lower bound on capacity  is,
\bea
C_{lower}(\Gamma,P)=2B\max_{P_{\ast},P_1,P_2}\left[\frac{h(f_{Y|S=g_1})+h(f_{
Y|S=g_2 } ) } { 2 } -h(\mathcal { N
}(0,NB))\right],\nonumber\\
\mbox{subject to }(1-\Gamma)P_{\ast}+\frac{\Gamma}{2}(P_1+P_2)\le P.
\eea
We plot $C_{lower}(\Gamma,P)$ as a function of $\Gamma$ for $P=N=1$, and
$g_1=0.01, g_2=1$ in 
Fig. \ref{examplefading}.
\begin{figure}[htbp] 
\begin{center}
\scalebox{0.5}{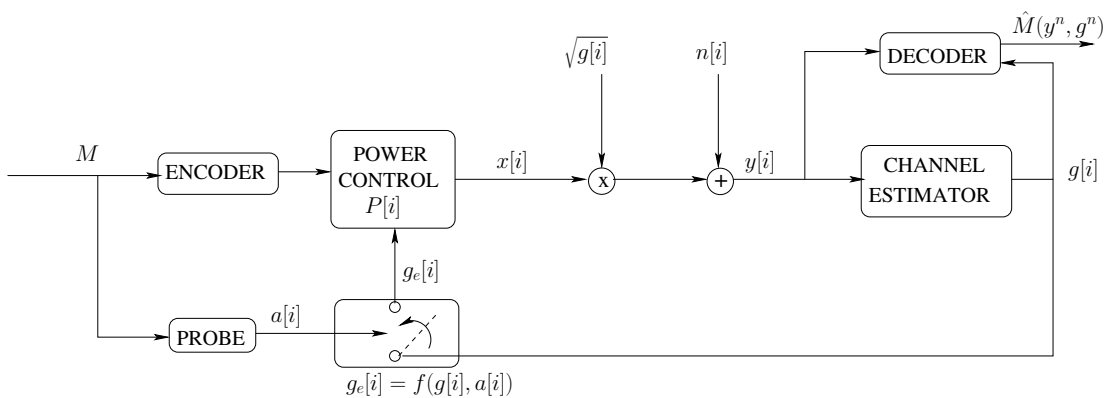}
\caption{Fading channels with encoder taking actions to acquire channel state
for adaptive power control.}
\label{fading}
\end{center}
\end{figure}
\begin{figure}[htbp]
\begin{center}
\includegraphics[scale=0.35]{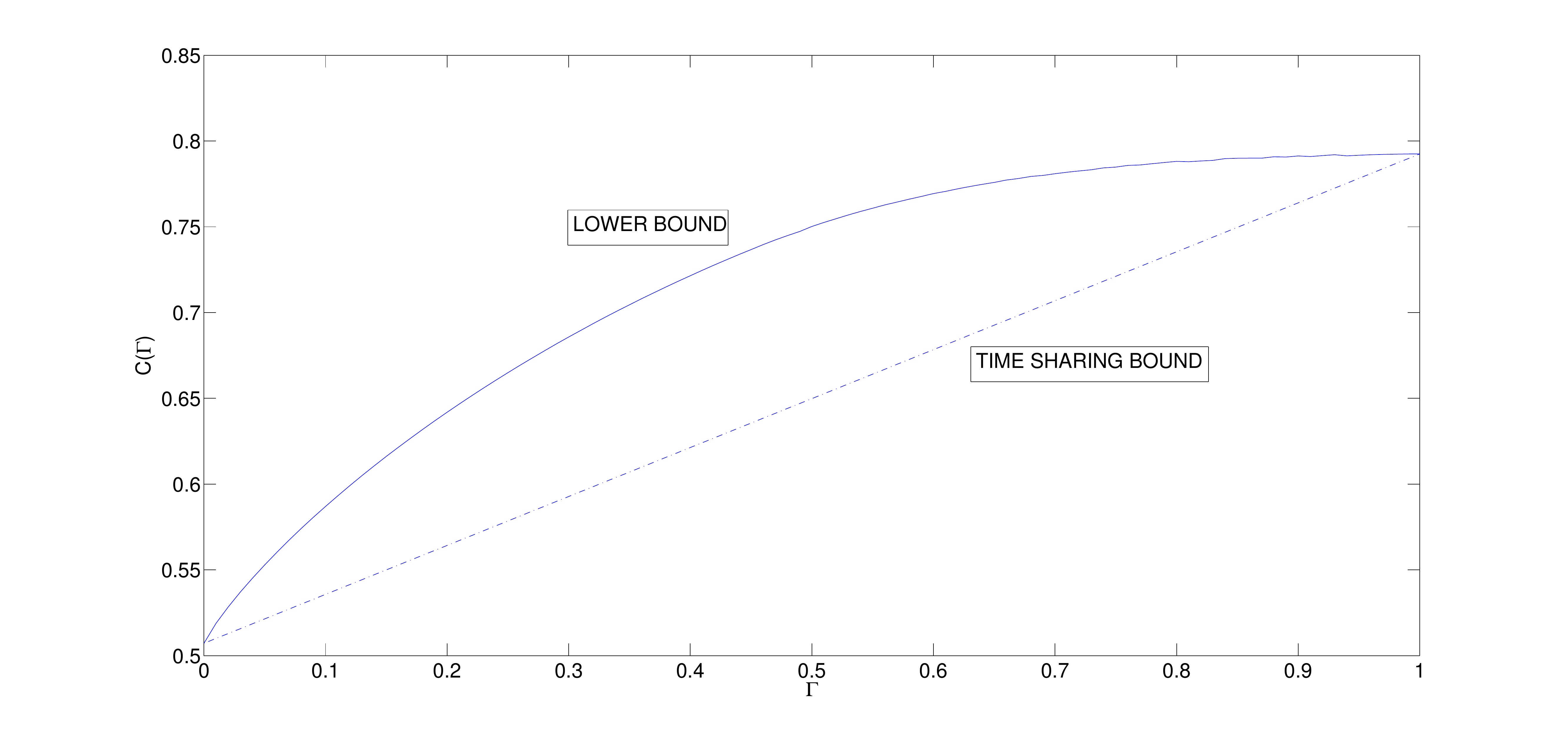}
\caption{Lower bound on fading channel communication system in Fig.
\ref{fading}. Time sharing is evidently highly sub-optimal.}
\label{examplefading}
\end{center}
\end{figure}

\section{Conclusion}
\label{conclusion}
In this work, we obtain 
\textquoteleft\textit{Probing Capacity}\textquoteright\ of systems which are
characterized as follows :
\begin{itemize}
 \item Channel is DMC with i.i.d states.
 \item Encoder takes \textit{costly} actions and probes the channel for channel
state information. This may be used causally or non-causally to generate
channel input symbols.
\item Decoder takes \textit{costly} actions and probes the channel to obtain
state information which is then used to construct message estimate.
\end{itemize}
We also worked out examples of discrete and continuous channels in cases where
only encoder probed the channel for states. We not only showed that a
naive time sharing scheme is strictly sub-optimal but also showed a pleasing
phenomenon (see Example 1. in Section \ref{example}) where one needs to
observe only a fraction of states to obtain maximum rate of transmission i.e.
rate when cost of state observation at encoder is not constrained.
\par
As directions of future work, following are important questions/conjectures
worth spending time and energy,
\begin{itemize} 
\item [1.] What if encoder actions depend on \textit{past sampled state},
i.e., 
$A_{e,i}=A_{e,i}(M,S_e^{i-1})$ for the case when partial state information is to
be used non-causally ? Can capacity be increased ?
\item [2.] What about probing capacity for channels with \textit{memory} ?
\item [3.] Does the Example 4 on \textquoteleft\textit{Learning
to write on a dirty paper}\textquoteright\ also support the pleasing phenomenon
when we can observe only a fraction of states and still achieve Costa's dirty
paper coding capacity, $C(P/N)$ ?
\item [4.] What if we take action \textit{to sample or not feedback} at
encoder or decoder for channels with memory ? 
\end{itemize}
Some of the results concerning sampling or not the feedback for finite state
channels (FSC) have been characterized in \cite{AsnaniHaimTsachyFeedback},
while the rest are under investigation.

\bibliography{fsc}

\begin{thebibliography}{10}
\providecommand{\url}[1]{#1}
\csname url@samestyle\endcsname
\providecommand{\newblock}{\relax}
\providecommand{\bibinfo}[2]{#2}
\providecommand{\BIBentrySTDinterwordspacing}{\spaceskip=0pt\relax}
\providecommand{\BIBentryALTinterwordstretchfactor}{4}
\providecommand{\BIBentryALTinterwordspacing}{\spaceskip=\fontdimen2\font plus
\BIBentryALTinterwordstretchfactor\fontdimen3\font minus
  \fontdimen4\font\relax}
\providecommand{\BIBforeignlanguage}[2]{{%
\expandafter\ifx\csname l@#1\endcsname\relax
\typeout{** WARNING: IEEEtran.bst: No hyphenation pattern has been}%
\typeout{** loaded for the language `#1'. Using the pattern for}%
\typeout{** the default language instead.}%
\else
\language=\csname l@#1\endcsname
\fi
#2}}
\providecommand{\BIBdecl}{\relax}
\BIBdecl

\bibitem{ShannonChannel}
C.~E. Shannon, ``Channels with side information at the transmitter,'' \emph{IBM
  J. Res. Dev.}, vol.~2, no.~4, pp. 289--293, 1958.

\bibitem{KuznetsovTsybakov}
A.~V. Kuznetsov and B.~S. Tsybakov, ``Coding in a memory with defective
  cells,'' \emph{Probl. Contr. and Inform. Theory.}, vol.~10, no.~2, pp.
  52--60, 1974.

\bibitem{GelfandPinsker}
S.~I. Gel'fand and M.~S. Pinsker, ``Coding for channel with random
  parameters,'' \emph{Problems of Control Theory}, vol.~9, no.~1, pp. 19--31,
  1980.

\bibitem{HeegardGamal}
C.~D. Heegard and A.~A. El~Gamal, ``On the capacity of computer memory with
  defects,'' \emph{IEEE Trans. Inf. Theor.}, vol.~29, no.~5, pp. 731--739,
  September 1983.

\bibitem{KeshetSteinbergMerhav}
G.~Keshet, Y.~Steinberg, and N.~Merhav, ``Channel coding in the presence of
  side information,'' \emph{Found. Trends Commun. Inf. Theory}, vol.~4, no.~6,
  pp. 445--586, 2007.

\bibitem{HaimTsachyVendor}
H.~H. Permuter and T.~Weissman, ``Source coding with a side information
  'vending machine' at the decoder,'' in \emph{ISIT'09: Proceedings of the 2009
  IEEE international conference on Symposium on Information Theory}.\hskip 1em
  plus 0.5em minus 0.4em\relax Piscataway, NJ, USA: IEEE Press, 2009, pp.
  1030--1034.

\bibitem{WynerZiv}
A.~D. Wyner and J.~Ziv, ``The rate-distortion function for source coding with
  side information at the decoder,'' \emph{IEEE Trans. Inform. Theory},
  vol.~22, pp. 1--10, 1976.

\bibitem{TsachyChannel}
T.~Weissman, ``Capacity of channels with action-dependent states,'' in
  \emph{ISIT'09: Proceedings of the 2009 IEEE international conference on
  Symposium on Information Theory}.\hskip 1em plus 0.5em minus 0.4em\relax
  Piscataway, NJ, USA: IEEE Press, 2009, pp. 1794--1798.

\bibitem{Kittichokechai}
K.~Kittichokechai, T.~Oechtering, M.~Skoglund, and R.~Thobaben, ``Source and
  channel coding with action-dependent partially known two-sided state
  information,'' in
  \emph{ISIT'10: Proceedings of the 2010 IEEE international conference on
  Symposium on Information Theory}, June 2010, pp. 629 --633.

\bibitem{Donoho}
D.~L. Donoho, ``Compressed sensing,'' \emph{IEEE Trans. Inform. Theory},
  vol.~52, pp. 1289--1306, 2006.

\bibitem{GoldsmithVaraiya}
A.~J. Goldsmith and P.~P. Varaiya, ``Capacity of fading channels with channel
  side information,'' \emph{IEEE Trans. Inform. Theory}, vol.~43, pp.
  1986--1992, Nov. 1997.

\bibitem{CsiszarKorner}
I.~Csiszar and J.~Korner, \emph{Information Theory: Coding Theorems for
  Discrete Memoryless Systems}.\hskip 1em plus 0.5em minus 0.4em\relax Orlando,
  FL, USA: Academic Press, Inc., 1982.

\bibitem{GamalKim}
A.~E. Gamal and Y.~H. Kim, ``Lecture notes on network information theory,''
  \emph{CoRR}, vol. abs/1001.3404, 2010.

\bibitem{CoverThomas}
T.~M. Cover and J.~A. Thomas, \emph{Elements of information theory}.\hskip 1em
  plus 0.5em minus 0.4em\relax New York, NY, USA: Wiley-Interscience, 1991.

\bibitem{Salehi}
M.~Salehi, ``Cardinality bounds on auxiliary variables in. multiple-user theory
  via the method of ahlswede and korner,'' Department of Statistics, Stanford
  University, Stanford, CA, Tech. Rep.~33, August 1978.

\bibitem{ZaidiVandendorpe}
A.~Zaidi, L.~Vandendorpe, and P.~Duhamel, ``Lower bounds on the capacity
  regions of the relay channel and the cooperative relay-broadcast channel with
  non-causal side information.'' jun. 2007, pp. 6005 --6011.

\bibitem{CostaDPC}
M.~Costa, ``Writing on dirty paper (corresp.),'' \emph{Information Theory, IEEE
  Transactions on}, vol.~29, no.~3, pp. 439 -- 441, may. 1983.

\bibitem{AsnaniHaimTsachyFeedback}
H.~Asnani, H.~H. Permuter, and T.~Weissman, ``To feed or not to feed back,'' in
  preparation.

\end{thebibliography}
\bibliographystyle{IEEEtran}

\appendices

\section{Concavity of Capacity in Cost}
\label{concavification}
We prove the concavity of cost constrained capacity in Theorem \ref{theorem4}
by concavification argument. Consider \textquoteleft
concavification\textquoteright\ of capacity in Theorem 
\ref{theorem4} as
\bea\label{cap3}
C^Q(\Gamma)=\max[I(U;Y,S_d|A_d,Q)],
\eea
where maximization is over all joint distributions of the form
\bea
&&P_{Q,S,A_d,U,A_e,S_e,X,Y,S_d}(s,a_d,u,a_e,s_e,x,y,s_d)\nonumber\\
&=&P_Q(q)P_S(s)P_{A_d|Q}(a_d|q)P_{U|A_d,Q}(u|a_d,q)\1_{\{a_e=g(u,a_d,q)\}}
\nonumber\\
&&P_{S_e|S,A_e}(s_e|s,
a_e)\1_{\{x=f(u,s_e,a_d,q)\}}P_{Y|X,S}P_{S_d|S,A_d}(s_d|s,a_d),
\eea
for some $P_Q,P_{A_d|Q},P_{U|A_d|Q},g,f$ such that
$\E[\Lambda(A_e,A_d)]\leq\Gamma$. Clearly $C^Q(\Gamma)\geq C(\Gamma)$. Left is
to prove $C^Q(\Gamma)\leq C(\Gamma)$.
\bea
I(U;Y,S_d|A_d,Q)&=&H(Y,S_d|A_d,Q)-H(Y,S_d|U,A_d,Q)\\
&\leq&H(Y,S_d|A_d)-H(Y,S_d|U,A_d,Q)\\
&=&I(U';Y,S_d|A_d),
\eea
where last inequality follows from the defining $U'=(U,Q)$. Proof is completed
by noting that the joint distribution of $(S,A_d,U',A_e,S_e,X,Y,S_d)$ is same as
that of $(S,A_d,U,A_e,S_e,X,Y,S_d)$.\newline

\section{Proof of Markov Chain $X_i-(S_{e,i},A_i)-S_i$}
\label{markov1}
Since
$X_i=X_i(M,S_e^n)$, it suffices to prove
 $(M,S_e^n)-(S_{e,i},A_i)-S_i$. We observe the joint distribution can be
factorized as,
\bea
P(M,A^n,S^n,S^n_e)&=&P(M)\prod_{i=1}^{n}P(S_i)P(A_i|M)P(S_{e,i}|S_i,A_i)\\
&=&\Phi_1(A^{n\backslash i},M,S_e^{n\backslash
i},S^{n\backslash i},A_i)\Phi_2(S_{i},S_{e,i},A_i)\\
&=&\Phi^{'}_1(A^{n\backslash i},M,S_e^{n},S^{n\backslash
i},A_i,S_{e,i})\Phi_2(S_{i},S_{e,i},A_i)
\eea
which implies the Markov Chain $(A^{n\backslash
i},M,S_e^{n},S^{n\backslash i})-(S_{e,i},A_i)-S_i$, which in turn
implies  $(M,S_e^{n})-(S_{e,i},A_i)-S_i$.

\section{Proof of Markov Chains in Theorem \ref{theorem4}}
\label{markov2}
We will prove the following markov chains,
\begin{enumerate}
\item[MC1] $U_i-A_{d,i}-S_i$.
\item[MC2] $A_{e,i}-(U_i,A_{d,i})-S_i$.
\item[MC3] $(S_{e,i},S_{d,i})-(S_i,A_{e,i},A_{d,i})-U_i$.
\item[MC4] $X_i-(U_i,S_{e,i},A_{d,i})-(A_{e,i},S_i,S_{d,i})$.
\item[MC5] $Y_{i}-(X_i,S_{i})-(U_i,A_{d,i},A_{e,i},S_{e,i},S_{d,i})$.
\end{enumerate}
MC3 and MC5 follow from the DMC assumption in problem definition. Now
for
the rest consider the induced probability distribution by the given encoding and
decoding scheme,
\bea
&&P_{M,A_e^n,S^n,S^n_e,X^n,Y^n,A_d^n,S^n_d}(m,a_e^n,s^n,s^n_e,x^n,y^n,a^n_d,
s^n_d
)\nonumber\\
&=&\frac{1}{\mathcal{M}}\1_{\{a_e^n=A_e^n(m)\}}\prod_{i=1}^{n}\1_{\{a_{d,i}=A_{d
,i}(y^{i-1})\}}P_{S}(s_i)P_{S_e,S_d|S
,A_e,A_d}(s_{e,i},s_{d,i}|s_i,a_{e,i},a_{d,i})
\nonumber\\
&&\times\prod_{i=1}^{n}\1_{\{x_{e,i}=X_{e,i}(m,s_e^i)\}}P_{Y|X,S}(y_i|x_i,
s_i)\label{induced1}.
\eea
Averaging over $(S^n_{i+1},S_{e,i}^n,X^n_{i},Y^n_i,S_{d,i}^n,A_{d,i+1}^n)$, we
get
\bea
&&P_{M,A_e^n,S^i,S^{i-1}_e,X^{i-1},Y^{i-1},A_d^{i-1},S^{i-1}_d}(m,a_e^n,s^i,s^{
i-1}_e,x^{i-1},y^{i-1},a^{i-1}_d,s^{i-1}_d
)\nonumber\\
&=&P_{S}(s_i)\times\frac{1}{\mathcal{M}}\1_{\{a_e^n=A_e^n(m)\}}\1_{\{a_{d,i}=A_{
d,i}(y^{i-1})\}}\prod_{j=1}^{i-1}\1_{\{a_{d,j}=A_{d,j}(y^{j-1})\}}P_{S}(s_j)P_{
S_e,S_d|S,A_e,A_d}(s_{e,j}s_{d,j}|s_j,a_{e,j},a_{d,j})\nonumber\\
&&\times\prod_{j=1}^{i-1}\1_
{\{x_{e,j}=X_{e,j}(m,s_e^j)\}}P_{Y|X,S}(y_j|x_j,s_j))\\
&=&\Phi_1(S_i)\Phi_2(M,A_e^n,S^{i-1},S^{i-1}_e,X^{i-1},Y^{i-1},A_d^{i},S^{i-1}
_d)\label{rel1}\\
&=&\Phi_1'(S_i,A_{d,i})\Phi_2(A_{d,i},U_i,X^{i-1})\label{MC1}.
\eea
Eq. (\ref{rel1}) implies $A_{d,i}$ is independent of $S_i$ while Eq.
(\ref{MC1}) implies markov chain $(U_i,X^{i-1})-A_{d,i}-S_i$ which in turn
implies
MC1. MC2 is straightforward as $U$ contains $A_e^n$. 
\par 
Now averaging over
$(S^n_{i+1},S_{e,i+1}^n,X^n_{i+1},Y^n_i,S_{d,i+1}^n,A_{d,i+1}^n)$ in Eq.
(\ref{induced1}) we obtain,
\bea
&&P_{M,A_e^n,S^i,S^{i}_e,X^{i},Y^{i-1},A_d^{i-1},S^{i-1}_d}(m,a_e^n,s^i,s^{i}_e,
x^{i},y^{i-1},a^{i-1}_d,s^{i-1}_d
)\nonumber\\
&=&P_{S}(s_i)P_{S_e,S_d|S,A_e,A_d}(s_{e,i},s_{d,i}|s_i,a_{e,i},a_{d,i}
)\nonumber\\
&&\times\frac{1}{\mathcal{M}}\1_{\{a_e^n=A_e^n(m)\}}\1_{\{x_{i}=X_{i}(m,
s_e^i)\}}\1_{\{a_{d,i}=A_{d,i}(y^{i-1})\}}\prod_{j=1}^{i-1}P_{S}(s_j)P_{S_e,
S_d|S ,
A_e,A_d}(s_{e,j},s_{d,j}|s_j,a_{e,j},a_{d,j})
\nonumber\\
&&\times\prod_{j=1}^{i-1}\1_{\{x_{e,j}=X_{e,j}(m,s_e^j)\}}P_{Y|X,S}(y_j|x_j,
s_j)\1_{\{a_{d,j}=A_{ d , j } (y^{j-1})\}}\\
&=&\Phi_1(S_i,S_{e,i},S_{d,i},A_{e,i},S_{e,i})\Phi_2(M,A_e^n,S^{i-1},S^{i-1}_e,
X^ { i } , Y^ { i-1 } ,
A_d^{i},S^{i-1}_d)\\
&=&\Phi'_1(S_i,A_{e,i},S_{d,i},U_i,S_{e,i},A_{d,i})\Phi'_2(U_i,S_{e,i},A_{d,i},
X^i , S^ {i-1}).
\eea
This implies the Markov Chain,
$(S^{i-1},X^i)-(U_i,S_{e,i},A_{d,i})-(S_i,A_{e,i},S_{d,i})$ which
implies MC4.

\end{document}